\documentclass[11pt]{article}
\usepackage[utf8]{inputenc}
\usepackage[T1]{fontenc}
\usepackage[letterpaper,top=1in,bottom=1in,left=1in,right=1in,marginparwidth=1.75cm]{geometry}
\usepackage[nottoc,numbib]{tocbibind}
\usepackage{graphicx}
\usepackage{authblk}
\usepackage{complexity}
\usepackage{bbm}	
\usepackage{amsmath}
\usepackage[ruled,vlined,linesnumbered,noresetcount]{algorithm2e}
\usepackage{amsfonts}
\usepackage{amsmath}
\usepackage[colorlinks=true, linkcolor=red, urlcolor=blue, citecolor=gray]{hyperref}
\usepackage{amssymb}

\usepackage{amsthm}
\usepackage{mathtools}
\usepackage{capt-of}
\theoremstyle{plain}
\usepackage{bbm}
\usepackage{bm}
\usepackage{xcolor}
\usepackage[shortlabels]{enumitem} 
\usepackage{comment}
\usepackage{tikz}   
\usepackage{xcolor}
\usepackage{cite}
\usetikzlibrary{decorations.pathreplacing}
\usepackage[title]{appendix}
\usepackage[framemethod=tikz]{mdframed}
\usepackage{titlesec}
\usepackage{todonotes}
\usepackage[misc]{ifsym}
\usepackage{amssymb}
\usepackage{amsmath}
\usepackage{algorithmic}
\usepackage{commath}
\usepackage[capitalise]{cleveref}

\DeclareMathOperator*{\argmin}{arg\,min}
\let\polylog\relax
\DeclareMathOperator*{\polylog}{\mathrm{polylog}}
\DeclareMathOperator*{\tr}{\mathrm{tr}}

\makeatletter
\newtheorem*{rep@theorem}{\rep@title}
\newcommand{\newreptheorem}[2]{%
	\newenvironment{rep#1}[1]{%
		\def\rep@title{#2 \ref{##1}}%
		\begin{rep@theorem}}%
		{\end{rep@theorem}}}

\newcommand{\undersim}[1]{\mathrel{\mathpalette\@undersim{#1}}}
\newcommand{\@undersim}[2]{%
  \vcenter{%
    \ialign{%
      ##\cr
      $\m@th#1#2$\cr
      \noalign{\nointerlineskip\kern.2ex}
      $\m@th#1\sim$\cr
      \noalign{\kern-.4ex}
    }%
  }%
}

\newcommand{\lsim}{\undersim{<}}
\makeatother

\newcommand{\eps}{\epsilon}
\newcommand{\wh}{\widehat}
\newcommand{\wt}{\widetilde}

\newcommand{\Z}{\mathbb{Z}}
\newcommand{\mcS}{\mathcal{S}}
\newcommand{\ta}{\tilde{a}}

\newcommand{\kstodo}[1]{\todo{KS: #1}}

\newcommand{\hltodo}[1]{\todo{HL: #1}}
\newcommand{\hltodoil}[1]{\todo[inline]{HL: #1}}

\newcommand{\Cam}[1]{\textcolor{blue}{Cam: #1}}

\let\R\relax
\newcommand{\R}{\mathbb{R}}
\DeclareMathOperator{\rank}{rank}
\newcommand{\diag}{\operatorname{diag}}
\let\norm\relax
\newcommand{\norm}[1]{\|#1\|}

\SetKw{Break}{break}

\newtheorem{theorem}{Theorem}
\newreptheorem{theorem}{Theorem}
\newtheorem{lemma}{Lemma}[section]
\newreptheorem{lemma}{Lemma}

\newtheorem{corollary}[theorem]{Corollary}
\newtheorem{claim}[lemma]{Claim}

\newtheorem{definition}[lemma]{Definition}
\newtheorem{remark}[lemma]{Remark}

\newtheorem*{claim*}{Claim}
\newtheorem*{proposition*}{Proposition}
\newtheorem*{lemma*}{Lemma}
\newtheorem*{problem*}{Problem}

  \usepackage{nth}
  \usepackage{intcalc}

{\makeatletter
	\gdef\xxxmark{%
		\expandafter\ifx\csname @mpargs\endcsname\relax 
		\expandafter\ifx\csname @captype\endcsname\relax 
		\marginpar{xxx}
		\else
		xxx 
		\fi
		\else
		xxx 
		\fi}
	\gdef\xxx{\@ifnextchar[\xxx@lab\xxx@nolab}
	\long\gdef\xxx@lab[#1]#2{{\bf [\xxxmark #2 ---{\sc #1}]}}
	\long\gdef\xxx@nolab#1{{\bf [\xxxmark #1]}}
}

\title{Toeplitz Low-Rank Approximation \protect\\ with Sublinear Query Complexity}

\author[1]{Michael Kapralov}
\author[2]{Hannah Lawrence}
\author[1]{Mikhail Makarov}
\author[3]{ \\ Cameron Musco}
\author[1]{Kshiteej Sheth}
\affil[1]{EPFL}
\affil[2]{MIT}
\affil[3]{University of Massachusetts Amherst}
\date{}

\begin{document}

\maketitle


\begin{abstract}
	We present a sublinear query algorithm for outputting a near-optimal low-rank approximation to any positive semidefinite Toeplitz matrix $T \in \R^{d \times d}$. In particular, for any integer rank $k \le d$ and $\epsilon,\delta > 0$, our algorithm makes $\wt O \left (k^2 \cdot \log(1/\delta) \cdot \text{poly}(1/\epsilon) \right )$ queries to the entries of $T$ and outputs a rank $\wt O \left (k \cdot  \log(1/\delta)/\epsilon\right )$ matrix $\wt T \in \R^{d \times d}$ such that $\| T - \wt T\|_F \leq (1+\epsilon) \cdot \norm{T-T_k}_F + \delta \norm{T}_F$. Here, $\norm{\cdot}_F$ is the Frobenius norm and $T_k$ is the optimal rank-$k$ approximation to $T$, given by projection onto its top $k$ eigenvectors.  $\wt O(\cdot)$ hides $\polylog(d) $ factors. 
	Our algorithm is \emph{structure-preserving}, in that the approximation $\wt T$ is also Toeplitz. A key technical contribution is a proof that any positive semidefinite Toeplitz matrix in fact has a near-optimal low-rank approximation which is itself Toeplitz. Surprisingly, this basic existence result was not previously known. Building on this result, along with the well-established off-grid Fourier structure of Toeplitz matrices [Cybenko'82], we show that Toeplitz $\tilde T$ with near optimal error can be recovered with a small number of random queries via a leverage-score-based off-grid sparse Fourier sampling scheme.
\end{abstract}





\clearpage 

 
\tableofcontents

\newpage



\section{Introduction.}

In scientific computing, engineering, and signal processing, highly structured matrices -- such as Toeplitz, circulant, hierarchical, and graph-structured matrices -- arise in many problems, often due to the discretization  of an underlying physical system. Such matrices are intensely studied and in many cases admit fast, even near-linear time algorithms for solving core linear algebraic problems. 

We investigate the possibility of \emph{sublinear} algorithms for such highly structured matrix classes, focusing  in particular on  Toeplitz matrices. A matrix $T \in \R^{d \times d}$ is Toeplitz if it is constant along its diagonals, i.e., for any $i,j,k,l \in [d]$, $i-j = k-l$ implies that $T_{i,j} = T_{k,l}$. Toeplitz matrices are ubiquitous in image and signal processing, where they arise as covariance matrices of stationary processes -- i.e., when the covariance between two measurements depends only on their distance in space or time \cite{Gray:2006ta}. They also arise in queuing theory, the solution of differential and integral equations, control theory, approximation theory, and beyond -- see \cite{Bunch:1985ti} for a review of their many applications. 
Reversing the rows of a Toeplitz matrix yields a Hankel matrix. These matrices also find wide applications in signal processing, system identification, and numerical computing \cite{Fazel:2013vw,Munkhammar:2017ud}.

A $d \times d$ Toeplitz matrix is specified by just $O(d)$ parameters, and Toeplitz matrices are a classic example of \emph{low displacement rank matrices}. They admit fast algorithms for many problems. A Toeplitz matrix can be multiplied by a vector in $O(d \log d)$ time via a fast Fourier transform. Toeplitz linear systems can be solved exactly in $O(d^2)$ time via the Levinson algorithm, and to high precision in $O(d \log^2 d)$ time using randomized methods \cite{XiaXiGu:2012,XiXiaCauley:2014}. A full Toeplitz eigendecomposition can be performed in $O(d^2 \log d)$ time  \cite{pan1999complexity}. However, little is known about  algorithms for Toeplitz matrices with provable correctness on general input instances and running time or query complexity scaling \emph{sublinearly} in the input size $d$.

\subsection{Our contributions.}

In this work, we  study sublinear query algorithms for symmetric positive semidefinite (PSD)  Toeplitz matrix low-rank approximation, which is a widely-studied problem \cite{Park:1999ta,Ishteva:2014wm} with applications to signal and image recovery \cite{Luk:1996ur,Cai:2016wh,Ongie:2017us}, signal direction of arrival estimation \cite{Krim:1996wm,Abramovich:1999vs}, and beyond. 
We show that a low-rank approximation to any PSD Toeplitz matrix with near-optimal error in the Frobenius norm can be computed using a sublinear number of queries to the input matrix.
In particular, letting $\norm{M}_F = \left (\sum_{i=1}^d \sum_{j=1}^d M_{ij}^2\right )^{1/2}$ denote the Frobenius norm of a matrix $M$, and letting $\wt O(\cdot)$ hide poly-logarithmic factors in the argument and in $d$, we have the following theorem: 


\begin{theorem}[Sublinear query Toeplitz low-rank approximation]\label{sublinearqueryalgo}
	There is a randomized algorithm that, given any PSD Toeplitz matrix $T\in \mathbb{R}^{d\times d}$, $\epsilon, \delta \in (0,1)$, and integer $k\leq d$, reads $\wt {O} \left (\frac{k^2 \log(1/\delta)}{\epsilon^{6}} \right )$ entries of $T$ and returns a symmetric Toeplitz matrix $\wt{T}$ with rank $\wt {O}\left (\frac{k  \log(1/\delta)}{\epsilon} \right )$ satisfying with probability at least $97/100$,
	\begin{equation}
		\|T-\wt{T}\|_F\leq (1+\epsilon) \|T-T_k\|_F + \delta \|T\|_F,
	\end{equation}
	where $\displaystyle T_k = \argmin_{B: \rank(B)\le k} \norm{T-B}_F$ is the best rank-$k$ approximation to $T$ in the Frobenius norm.
\end{theorem}

Theorem \ref{sublinearqueryalgo}  gives a near-relative error approximation, up to an additive $\delta \norm{T}_F$. However, the dependence on $\delta$ in the sample complexity and output rank is just logarithmic. Thus, this additive term is comparable to the error that would be introduced in any practical algorithm due to round-off error. To the best of our knowledge, ours is the first sublinear query algorithm for Toeplitz low-rank approximation achieving near-relative error. See Section \ref{sec:prior} for a detailed comparison to prior work. 


\smallskip

\noindent{\bf Structure-preservation and bicriteria approximation.} Observe that  Theorem \ref{sublinearqueryalgo} outputs a low-rank Toeplitz matrix $\wt T$, even though the optimal low-rank approximation $T_k$ will in general not be Toeplitz. 
This \emph{structure-preserving low-rank approximation} is desirable in many applications \cite{chu2003structured,Cai:2016wh}. It allows fast methods for Toeplitz matrices to be applied to $\wt T$ itself, and is useful when the Toeplitz structure of the approximation has a physical meaning, in applications such as direction of arrival estimation \cite{Krim:1996wm}. As we will discuss, $\wt T$'s Toeplitz structure is also key to achieving sublinear query complexity -- it allows us to recover $\wt T$ via query efficient sparse Fourier transform techniques.

Since $\wt T$ is Toeplitz, the approximation bound of Theorem \ref{sublinearqueryalgo} is \emph{bicriteria} -- i.e., $\wt{T}$ has rank larger than the input rank $k$. This is necessary since as $\epsilon,\delta$ go to zero, $\| T-\wt T \|_F$ becomes arbitrarily close to $\norm{T-T_k}_F$. However, one can find simple examples of $T$ and $k$ where there is a fixed gap between $\norm{T-T_k}_F$ and the error of the best rank-$k$ Toeplitz  approximation to $T$. See Figure \ref{fig:bicriteria}. It is an interesting open question if a sublinear query  algorithm for Toeplitz low-rank approximation with the same near-relative error bound as Theorem \ref{sublinearqueryalgo} can be achieved without any bicriteria approximation, by abandoning the structure-preserving guarantee and allowing $\wt T$ to be non-Toeplitz.

\begin{figure}[h]
	$$T = \begin{bmatrix} 2 & 1 & 0 \\ 1 & 2 & 1 \\ 0 & 1 & 2 \end{bmatrix} \hspace{3em} T_1 = \frac{2+\sqrt{2}}{4} \cdot \begin{bmatrix} 1 & \sqrt{2}  & 1  \\ \sqrt{2} & 2 & \sqrt{2}  \\ 1 & \sqrt{2}  & 1  \end{bmatrix} \hspace{3em} T_{1,toep} = \frac{10}{9} \cdot \begin{bmatrix} 1 & 1 & 1 \\ 1 & 1 & 1 \\ 1 & 1 & 1 \end{bmatrix}$$
	\caption{Example of a $3 \times 3$ positive semidefinite Toeplitz matrix $T$ whose best rank-$1$ approximation $T_1 = \argmin_{B:\rank(B) \le k} \norm{T-B}_F$ differs from its best rank-$1$ Toeplitz approximation $T_{1,toep} = \argmin_{B:\rank(B) \le k, \, \text{B is Toeplitz}} \norm{T-B}_F$. We can check that $\norm{T-T_{1,toep}}_F - \norm{T-T_1}_F \approx 0.1271$. Thus, any Toeplitz low-rank approximation to $T$ achieving error $(1+\epsilon) \norm{T-T_1}_F$ for small enough $\epsilon$ must have rank $> 1$. I.e., it must be a bicriteria approximation, as in Theorem \ref{sublinearqueryalgo}. $T_1$ is computed via projection onto $T$'s top eigenvector. $T_{1,toep}$ is computed by observing that any rank-$1$ Toeplitz matrix must have all entries of the same magnitude. Thus, since $T$ has all positive entries, the optimal approximation is the all ones matrix scaled by the mean entry in $T$.}\label{fig:bicriteria}
\end{figure}

\noindent{\bf Runtime.} In Theorem \ref{sublinearqueryalgo} we focus on query complexity rather than runtime, and  our algorithm requires time exponential in $\wt O(k/\epsilon)$ to identify $\wt T$ via a brute-force-search-based off-grid sparse Fourier transform algorithm. This is not sublinear time unless $k$ is sublogarithmic in $d$. However, we conjecture that sublinear runtime is possible by adapting efficient off-grid sparse Fourier transform techniques \cite{chen2016fourier} to  find $\wt T$. 

\medskip

\noindent{\bf Existence proof.} A key challenge in proving Theorem \ref{sublinearqueryalgo} is  to demonstrate that a Toeplitz $\wt T$ achieving a near-optimal Frobenius norm low-rank approximation even \emph{exists}. Surprisingly, this was not previously known. We show
\begin{theorem}[Existence of near-optimal Toeplitz low-rank approximation]\label{frobeniusexistence}
	For any PSD Toeplitz $T\in \mathbb{R}^{d\times d}$ , $\eps,\delta \in (0,1)$ and  integer $k\leq d$ there exists a symmetric Toeplitz matrix $\wt{T} \in \R^{d \times d}$ of rank $\wt{O}\left (\frac{k \log(1/\delta)}{\eps}\right )$ such that
	\begin{align*}
		\|T-\wt{T}\|_F \leq (1+\eps)  \|T-T_k\|_F + \delta\|T\|_F,
	\end{align*}
	where $\displaystyle T_k = \argmin_{B: \rank(B)\le k} \norm{T-B}_F$ is the best rank-$k$ approximation to $T$ in the Frobenius norm.
\end{theorem}
We  prove a similar result for spectral norm low-rank approximation -- see Theorem \ref{spectralexistence}. 
One natural approach to proving Theorem~\ref{frobeniusexistence} (and Theorem~\ref{spectralexistence}) is  to note that, intuitively, if $T$ is close to low-rank, its optimal low-rank approximation $T_k$ -- given by projection onto its top $k$ eigenvectors -- should be nearly Toeplitz. Of course, rounding this matrix to be Toeplitz (by replacing the entries on each diagonal by their average) can increase its rank, but one could plausibly  try to bound the increase in rank. This, however, would not naturally lead to a sublinear query algorithm for recovering a low rank approximation. We take a different approach, exploiting the classical Fourier structure of Toeplitz matrices, namely the Vandermonde decomposition~\cite{cybenko1982moment}. In essence, instead of rounding the best rank-$k$ approximation to Toeplitz, we round the matrix $T$ itself to a rank-(almost)$k$ matrix {\em in Fourier domain}, thereby preserving the Toeplitz property throughout. This lets us use a number of powerful ideas from the literature on recovering Fourier sparse functions from few measurements, and ultimately leads to a sublinear query algorithm -- see  Section~\ref{sec:overview} for a more detailed overview.

We believe that Theorem \ref{frobeniusexistence} is of interest beyond its application to proving Theorem \ref{sublinearqueryalgo}. It remains an open question if the $\delta \norm{T}_F$ term can be removed. Although we stated our guarantees as applying to exactly Toeplitz matrices, they extend fairly directly to near-Toeplitz input matrices. For example, our main existence result in Theorem \ref{frobeniusexistence} easily extends to arbitrary matrices which are close in Frobenius norm to a PSD Toeplitz matrix, by the triangle inequality. Moreover, the sublinear query complexity algorithm of Theorem \ref{sublinearqueryalgo} will extend to any matrix whose first column is close to the first column of a PSD Toeplitz matrix, according to the weighted $\ell_2$ norm, as defined in Claim \ref{frobtoeuclid}. It would be interesting to extend this further, e.g. to any matrix that is close to a PSD Toeplitz matrix in the Frobenius norm. Additionally, while $\wt T$ must have a rank larger than $ k$ when $\epsilon,\delta$ are sufficiently small (see Figure \ref{fig:bicriteria}), identifying the minimum rank required to achieve the given error bound is also a very interesting problem.

\subsection{Related work.}\label{sec:prior} 

Significant prior work in  numerical linear algebra and signal processing  has studied the low-rank approximation of Toeplitz and, relatedly, Hankel matrices \cite{Luk:1996ur,Park:1999ta,Ishteva:2014wm,Cai:2016wh,Ongie:2017us,Krim:1996wm}. In signal processing applications,  the input Toeplitz matrix $T$ is often a covariance matrix and thus positive semidefinite. 

Due to the fact that a $d \times d$ Toeplitz matrix can be multiplied by a vector in $O(d \log d)$ time via fast Fourier transform, a near-optimal low-rank approximation can be computed in near-linear time. In particular, Shi and Woodruff present an algorithm that outputs rank-$k$ $\wt T$ with $\norm{T-\wt T}_F \le (1+\epsilon) \norm{T-T_k}_F$ in $\wt O(d + \poly(k/\epsilon))$ time \cite{Shi:2019ud}. Significant other work focuses on computing a near-optimal low-rank approximation that preserves Toeplitz structure, as we do in Theorem \ref{sublinearqueryalgo}. Unlike unconstrained low-rank approximation, where the optimal solution  can be computed directly via eigendecomposition, no simple characterization of the optimal structure-preserving Toeplitz low-rank approximation is known \cite{chu2003structured}. Computing such an optimal approximation in polynomial time remains open outside the special cases of $k= 1$ and $k = d-1$ \cite{chu2003structured,Knirsch:2021ve}. Practical heuristics apply a range of techniques, based on  convex relaxation \cite{Fazel:2013vw,Cai:2016wh,Ongie:2017us}, alternating minimization \cite{chu2003structured,Wen:2020ub}, and sparse Fourier transform \cite{Krim:1996wm}.  Observe that our main result, Theorem \ref{sublinearqueryalgo} directly gives a near-relative error bicriteria approximation algorithm for the optimal  structure-preserving Toeplitz low-rank approximation problem, since $\norm{T-T_k}_F = \min_{B: \rank(B) \le k} \norm{T-B}_F \le \min_{B: \rank(B)\le k, B\text{ is Toeplitz}} \norm{T-B}_F$.


Several works also investigate sublinear query algorithms for Toeplitz matrices \cite{Abramovich:1999vs,Chen:2015wz,Qiao:2017tp,Lawrence:2020ut,eldar2020toeplitz}. In the signal processing community, these algorithms are often framed in terms of \emph{sparse array methods}, which can be thought of as reading a small principal submatrix of $T$ from which an approximation to the full matrix can be recovered.
%
%
%
%
%
%
%
%
Most closely related to our work is that of Eldar, Li, Musco, and Musco \cite{eldar2020toeplitz}, which focuses on approximating a PSD Toeplitz matrix $T \in \R^{d \times d}$ given samples from a $d$-dimensional Gaussian distribution with covariance $T$. They focus on minimizing both the number of samples taken from the distribution, as well as the number of entries read from each $d$-dimensional sample. For nearly low-rank $T$, they present an algorithm that estimates $T$ via a low-rank approximation of the sample covariance matrix, computed from a small number of randomly selected entries of that matrix. One can check that this algorithm can be  directly applied to $T$ itself. It gives $\wt O(k^2 \cdot \log(1/\delta))$ query complexity, and outputs $\wt{T}$ with rank $O(k)$ achieving error 
$$ \|T-\wt{T}\|_F \le C \cdot \sqrt{ \tr(T) \cdot \left (\norm{T- T_k}_2 + \frac{\tr(T-T_k)}{k} \right ) } + \delta \norm{T}_F,$$
for some constant $C \gg 1$. The above guarantee is non-standard, and due to the $\tr(T)$ term, it can be much weaker than a relative error low-rank approximation. For example, $\tr(T) \gg \norm{T-T_k}_F$ if the  eigenvalues of $T$ decay quickly, which is typically the case in settings where low-rank approximation is applied. 
In fact, we can observe that the above 
guarantee is strictly weaker than that of Theorem \ref{sublinearqueryalgo}: 
letting the eigenvalues of $T$ be denoted $\lambda_1(T)  \ge \ldots \ge \lambda_n(T) \ge 0$, 
$$\tr(T) \cdot \norm{T-T_k}_2 = \sum_{i=1}^n \lambda_i(T) \cdot \lambda_{k+1}(T) \ge \sum_{i=k+1}^n \lambda_i(T)^2 = \norm{T-T_k}_F^2.$$ 
Finally, the algorithm of  \cite{eldar2020toeplitz} outputs $\wt T$ which is not structure-preserving as in Theorem \ref{sublinearqueryalgo}. Note that in this sampling model of \cite{eldar2020toeplitz}, it is not possible to achieve a ‘near-relative’ error guarantee like we do. For example, even if the matrix is exactly rank-$k$, the sampling will lead to error approximately $\delta \|T\|_2$, with a polynomial in $\delta$ number of samples. In contrast, our algorithm (with exact access to the entries of $T$) gives error $\delta\|T\|_F$ with just a logarithmic dependence on $\delta$. Consider e.g., the setting when $T$ is just a matrix with every entry equal to $\alpha$ for $\alpha = \Theta (1)$. Approximating $\alpha$ to error $\pm \delta$ will require $1/\delta^2$  ‘vector samples’ in the model of \cite{eldar2020toeplitz}, for any algorithm.

However, our algorithm is indeed robust to noise, and so a similar guarantee with some additional additive error is achievable in the setting of \cite{eldar2020toeplitz}. One can take $\poly(d,1/\epsilon)$ vector samples and then using e.g., Claim 2.2 of \cite{eldar2020toeplitz}  one can argue that the sample covariance matrix will approximate the true covariance matrix, in that its first column will be close to the first column of the true matrix in the weighted $\ell_2$ norm, as defined in Claim \ref{frobtoeuclid}. We can then directly apply our algorithm to this sample covariance matrix. An interesting open problem here is to improve the vector sample complexity to just depend on $k$, and in general, to explore vector/entrywise sample complexity in more depth, as is done in \cite{eldar2020toeplitz}.

Beyond Toeplitz matrices, significant recent work has focused on  sublinear time low-rank approximation algorithms for other structured matrix classes. This includes positive semidefinite matrices \cite{musco2017recursive,Bakshi:2020tl}, distance matrices \cite{Bakshi:2018ul,Indyk:2019vy}, and kernel matrices \cite{musco2017recursive,Yasuda:2019vf,Ahle:2020vj}.

\subsection{Technical overview.}\label{sec:overview}

The main results of the paper are the proof of the existence of a near optimal low-rank Toeplitz approximation and the algorithm to recover it, presented in \Cref{sec:existence_main,algorithm} respectively. In this section we give an overview of the techniques used to achieve both results.

We start by introducing some notation. From classical works on Toeplitz matrices \cite{cybenko1982moment} it is known that any PSD Toeplitz matrix admits a \emph{Vandermonde decomposition} $T = F_S D F^*_S$, where $D$ is a diagonal matrix with positive entries and $F_S$ is a Fourier matrix. The Vandermonde decomposition is central to our work. We define Fourier matrices formally now.  

\begin{definition}[Frequency vector]\label{freqvector}
	For any frequency $f \in \mathbb{C}$, we define the {\em frequency vector} $v(f) \in \mathbb{C}^d$ as the column vector $[1, e^{2\pi i f}, \ldots, e^{2\pi i f (d-1)}]$.
\end{definition}
\begin{definition}[Symmetric Fourier matrix]\label{def:symmetric_fourier_matrix}
	For any set $S=\{f_1,f_2,\ldots, f_s\}\subset [0,1/2]$, let $F_S\in \mathbb{C}^{d\times 2s}$ be the Fourier matrix defined by the frequencies in $S$ as $F_{S} = [F_{+S};F_{-S}]$. 
	For every $j\in \{1, 2,\ldots, s\}$ the $j$-th column of $F_{+S}$ is $v(f_{j})$. Similarly, for every $j\in \{1, 2,\ldots, s\}$ the $j$-th column of $F_{-S}$ is $v(-f_{j})$.
\end{definition}

The formal statement of the Vandermonde decomposition is as follows.
\begin{theorem}[Vandermonde decomposition, Corollary 1 of \cite{cybenko1982moment}]\label{thm:vandermonde}
	Any real-valued PSD Toeplitz matrix $T \in \R^{d\times d}$ of rank $r$ can be expressed as $F_SDF_S^{*}$ where $F_S \in \mathbb{C}^{d\times r}$ is a symmetric Fourier matrix with the set of frequencies $S\subseteq [0,1/2]$ satisfying $|S|=r/2$, and $D\in \mathbb{R}^{r\times r}$ is a diagonal matrix with $r$ positive entries. 
	Moreover, for any $f\in S$ the values in $D$ corresponding to columns $v(f),v(-f)$, which we refer to as the weights of $f,-f$ respectively, are identical. Thus $D$ is uniquely defined by values $\{a_f\}_{f\in S}$.
\end{theorem}
\textbf{A basic Fourier-based approach and why it fails.} One approach to show the existence of a near optimal low-rank approximation, which itself is Toeplitz, is to show that retaining only the top (nearly) $k$ entries in the diagonal matrix $D$ in the Vandermonde decomposition $T=F_SDF_S^{*}$ and zeroing out the rest would give a near optimal low-rank approximation. This approach is natural, as this operation trivially preserves the Toeplitz structure. However, to formally argue that such an approach works would require us to relate the Vandermonde decomposition to the eigendecomposition of $T$. This is challenging in general because the Fourier matrix $F_S$ could potentially be highly ill-conditioned \cite{moitra2015super}, whereas the eigenvector matrix of $T$ has condition number $1$. Therefore, these decompositions could be very far from each other.\\
\\
\textbf{The special case of circulant matrices.} For the special but important case of circulant matrices, as defined below, relating these two decompositions turns out to be much easier.
\begin{definition}\label{def:circulant}
	A Toeplitz matrix $C \in \R^{d\times d}$ is called {\em circulant} if there exists a vector $c \in \R^d$ such that $C_{i, j} = c_{i + j - 1 \mod d}$.
\end{definition}
Essentially, a circulant matrix $C \in \R^{d \times d}$ is a matrix composed of all cyclic permutations of some vector $c \in \R^d$. Any circulant matrix is a Toeplitz matrix. A Vandermonde decomposition of a symmetric circulant matrix has the property that all of its frequencies are multiples of $1/d$~\cite{cybenko1982moment}. Thus the $F_S$ matrix is actually the discrete Fourier transform matrix, and its frequency vectors are orthogonal. In particular, the Vandermonde and eigendecompositions are identical in this special case, and we can easily argue that the best rank-$k$ approximation to a symmetric circulant matrix itself is circulant! This is because the best rank-$k$ approximation is given by retaining the top $k$ elements of $D$ in $F_SDF_S^{*}$, which has the circulant property by definition.\\
\\
\textbf{Toeplitz matrices.} In a similar spirit to the easy case of circulant matrices mentioned above, for a general Toeplitz matrix $T$, it is natural to ask whether it is true that if the Vandermonde decomposition of $T$ contains $k$ frequencies with large corresponding values, then the number of large eigenvalues of $T$ is also $\Omega(k)$. This is in general not the case. For example, if $S = \{f, -f, f + \eps, -f -\eps \ldots, f + (k/2) \eps, -f - (k/2) \eps \}$ for some $f \in [0, 1]$, $D$ is the identity matrix and $\eps \to 0$, then  both $F_S$ and $T$ tend to a rank $2$ matrix, with the third eigenvalue tending towards $0$. This happens because when the frequencies are close, their frequency vectors are highly correlated. Thus, for the purposes of our analysis it makes sense to consider a group of close frequencies as a single entity, and this observation motivates our proof plan.

We first analyze the case when all frequencies in $S$ are `close'  (this analysis is presented in \Cref{case1}), then analyze the interaction of `clusters' of frequencies in \Cref{case2} and \Cref{case3}, and finally derive the main results of \Cref{sec:existence_main}, namely \Cref{spectralexistence} and \Cref{frobeniusexistence}, in \Cref{finalproof}.  \Cref{spectralexistence} and \Cref{frobeniusexistence} state that for any $k$, there exists a symmetric Toeplitz matrix $\wt{T}$ of rank \emph{almost} $k$ which is \emph{almost} as good a low-rank approximation to $T$ as $T_k$, the best rank-$k$ approximation to $T$, in the spectral or Frobenius norm, respectively. 
We now give a high-level overview of the organization of \Cref{sec:existence_main} and constituent proofs.\\ 
\\
\textbf{Clustered case.} In Subsection \ref{case1}, we consider the case where the distance between each pair of frequencies in $S$ is at most $\Delta$ for some $\Delta\in [0, 1/2]$ (we ultimately choose $\Delta=O(1/d)$). This case is referred to as the \emph{clustered} case, formalized in the definition below.
\begin{definition}\label{clustered_case-mod}
	For $f^*\in [0, 1/2]$ and $\Delta\in [0, 1/2]$ we say that a Toeplitz matrix $T$ with Vandermonde decomposition $T=F_SDF_S^{*}$ is $(f^*,\Delta)$-clustered if $S\subset [f^*-\Delta,f^*+\Delta]$ for some $f^*\in [0,1/2]$ and $\Delta\geq 0$. Thus we can write any $f \in S$ as $f = f^*+r_f$ for some  $r_f\in [-\Delta, +\Delta]$. 
\end{definition}
Intuitively, in this case all the columns (i.e. frequency vectors) of $F_S$ are nearly identical, so we would expect $F_S$ (and thus $T$) to be close to an almost constant rank matrix. Furthermore, we can even show that this matrix $\wt{T}$ is Toeplitz itself, as specified by the following lemma.
\begin{lemma}\label{lem:clusteredcase}
	There exists a universal constant $C_1>0$ such that given any symmetric PSD Toeplitz  $T=F_SDF_S^*$ that is $(f^*,\Delta)$-clustered for some $f^*\in [0,1/2]$ and $\Delta \leq 1/d$ and  $0<\varepsilon,\delta,\gamma<1$ satisfying $\gamma\leq \varepsilon/(tr(T)2^{C_1\log^7 d})$ the following conditions hold.  There exists a symmetric Toeplitz matrix $\wt{T}=F_{\wt{S}}\wt{D}F_{\wt{S}}^{*}$ of rank at most $O(\ell)$ for $\ell=O(\log d + \log(1/\delta))$ such that 
	\begin{enumerate}
		\item $\wt{S}=\{f^*+j\gamma\}_{j=1}^{\ell+1}\cup \{f^*-j\gamma\}_{j=1}^{\ell+1}$ 
		and $\wt{D}$ is a diagonal matrix such that for any $f\in \wt{S}$ the weights corresponding to $f$ and $-f$ are identical. 
		\item $\|T-\wt{T}\|_{F} \leq \delta (\sum_{f\in S}a_{f}) +\varepsilon d$.
	\end{enumerate} 
\end{lemma}

\begin{remark}
	The value of the diagonal elements in $\widetilde{D}$ depends on $\varepsilon$, and in particular, the value of the diagonal elements would go to infinity as $\varepsilon$ goes to zero. However, we do not state this tradeoff explicitly in the lemma, because we apply it later in Section 4 to obtain leverage score upper bounds on Fourier sparse functions based on the work of \cite{chen2016fourier}. These bounds do not depend on the magnitude of the coefficients in the function, but only on its sparsity (as long as the coefficients are finite). Thus $\varepsilon$ can be set to any strictly positive number of one's choice.
\end{remark}

Note that the error term in point 2 of the above theorem is negligible, as the dependence of the rank of $\wt{T}$ on $1/\delta$ is logarithmic. The proof of this result uses tools from polynomial approximation and is achieved in three steps. First, we observe that since the $t$-th entry in the first column of $T$ is the linear combination of complex exponentials of form $e^{2 \pi i (f^* + r_f) t}$ where $r_f \leq \Delta$,  this column can be well approximated by a sum of polynomials of the form $p_j(t) = e^{(2\pi i f^{*}t)}\sum_{m=0}^l a_f \frac{(2 \pi i r_ft)^m}{m!}$ instead through the Taylor series. We then show that this sum of polynomials can be approximated by another polynomial on complex exponentials of the form $\wt{p}(t) =e^{(2\pi i  f^{*}t)} \sum_{m=-l - 1}^{l + 1} \alpha_m (e^{2 \pi i m \gamma t})$, from which we can then finally construct the matrix $\wt{T}$. The details are presented in \Cref{case1}.\\
\\
\textbf{Relating Vandermonde and eigenvalue decompositions.} 
We now discuss the high-level strategy of how we relate the Vandermonde and eigenvalue decompositions for general Toeplitz matrices. Central to our proof is the notion of a \emph{bucket}. We divide the interval $[0,1/2]$ into $d/2$ equal-sized sub-intervals, and refer to the frequencies of $F_S$ (recall $T=F_S D F_S^{*}$ is the Vandermonde decomposition of $T$) in each sub-interval as a \emph{bucket}. Each bucket corresponds to a group of close frequencies, corresponding to the {\bf clustered case} discussed above. The \emph{weight} of a bucket is the sum of the coefficients in $D$ with corresponding frequencies in that bucket. Using this notion, our first structural result that relates the two decompositions is as follows\footnote{Note that the theorem below is a natural relaxation of the claim that $k$ frequencies with coefficients at least $\lambda$ in the Vandermonde decomposition of a circulant matrix imply at least $k$ eigenvalues of value at least $\lambda$.}:

\begin{lemma}\label{heavybucketlemma}
	For every $\lambda,k>0$, if $T$ has at least $k$ buckets with weight at least $\lambda$, then $T$ has at least $k/\log^3 d$ eigenvalues that have value $\Omega(d\lambda/\log d)$. 
\end{lemma}
At a high level, this result follows from first showing that the interaction between frequencies, measured as an inner product $v(f_1)^* v(f_2)$ of their frequency vectors, has a harmonic decay as a function of the distance between $f_1$ and $f_2$. Then, we subsample the set of all heavy buckets to ensure that interaction between any remaining pair of buckets is small. Finally, we use a strengthening of Gershgorin's circle theorem for block matrices to guarantee that the contribution to each eigenvalue is dominated by only one bucket.

In Subsection \ref{case3}, we apply the same tools to show a complementary result\footnote{Again, note that the theorem below is a natural relaxation of the claim that if all frequencies in the Vandermonde decomposition of a  circulant matrix  have coefficients bounded by $\lambda$, then the spectral norm of $T$ is bounded by $d \lambda$.}:
\begin{lemma}\label{lightbucketlemma}
	For every $\lambda\geq 0$, if all buckets of $T$ have weight at most $\lambda$, then $\|T\|_2\leq O(d\lambda\log d)$.
\end{lemma}
In other words, a uniform bound on bucket weight implies a uniform bound over eigenvalues. Combined with the \Cref{heavybucketlemma}, this implies that the bucket structure characterizes the eigenvalue structure in some sense, up to polylogarithmic precision. These structural statements can also be seen as providing fine grained insights into the eigenvalue structure of arbitrary off-grid Fourier matrices beyond just condition number bounds \cite{moitra2015super}. This completes the overview of our approach of relating the potentially ill-conditioned Vandermonde decompostion and the eigenvalue decomposition for general Toeplitz matrices.\\
\\
\textbf{Putting it together: Toeplitz low-rank approximation in Frobenius and spectral norm.} Finally, Subsection \ref{finalproof} uses the structural statements of Subsections \ref{case2} and \ref{case3} to prove Theorems \ref{spectralexistence} and \ref{frobeniusexistence}.

We first prove Theorem \ref{spectralexistence}, the high level idea of which is as follows. Consider the Toeplitz matrix $\wt{T}$ obtained by taking only those buckets of weight at least $\wt{\Omega}(\lambda_{k+1}(T)/d)$, which we will call \emph{heavy}. By \Cref{lightbucketlemma}, $\|T-\wt{T}\|_2 \leq \wt{O}(\lambda_{k+1}(T)) \leq \wt{O}(1)\|T-T_k\|_2$. Moreover, by \Cref{heavybucketlemma} the matrix $\wt{T}$ cannot contain more than $\wt{O}(k)$ heavy buckets. Each bucket corresponds to a Toeplitz matrix with clustered frequencies, so replacing each bucket with its $\wt{O}(1)$-rank approximation by \Cref{lem:clusteredcase} incurs the additive error of $\delta \|T\|_F$ and finishes the proof. 

In general, however, a spectral norm low-rank approximation guarantee does not imply a Frobenius norm low-rank approximation guarantee. The main idea behind Theorem \ref{frobeniusexistence} is to consider the $\wt{T}$ obtained by taking a few more buckets than in the proof of Theorem \ref{spectralexistence}. Doing so ensures that $T$ is even closer to $\wt{T}$, such that even the top $k$ eigenvalues of $T-\wt{T}$ can be bounded in terms of the $d-k$ smallest eigenvalues of $T$. 
In particular, let $\wt{T}$ be the Toeplitz matrix obtained by taking only those buckets of weight at least $\wt{\Omega}(\lambda_{\wt{O}(k/\eps)}(T)/d)$. Again by \Cref{lightbucketlemma}, $\wt{T}$ cannot contain more than $\wt{O}(k/\eps)$ heavy buckets. \Cref{heavybucketlemma} bounds each of the top $k+1$ eigenvalues of $T-\wt{T}$ by $\wt{O}(\lambda_{\wt{\Omega}(k/\eps)}(T))$, so their contribution to $\|T-\wt{T}\|_F^2$ is at most $k\wt{O}(\lambda_{\wt{\Omega}(k/\eps)}(T)^2)\leq \eps \sum_{i=k+1}^{\wt{O}(k/\eps)} \lambda_i(T)^2\leq \eps\|T-T_k\|_F^2$. On the other hand, since $T-\wt{T}\preceq T$, we can bound any eigenvalue of $T-\wt{T}$ except the top $k+1$ by the corresponding eigenvalue of $T$. Thus the remaining eigenvalues' contribution to $\|T-\wt{T}\|_F^2$ can be bounded by $\sum_{i=k+2}^{d}\lambda_{i}^2(T)=\|T-T_k\|_F^2$. Overall this results in $\|T-\wt{T}\|_F \leq (1+\eps)\|T-T_k\|_F$. Replacing
each bucket in $\wt{T}$ with the $\wt{O}(1)$-rank matrix from \Cref{lem:clusteredcase} incurs the additive error and completes the proof. 
\\
\\
\textbf{Sublinear query algorithm.} Now we present the main ideas behind our recovery algorithm, \Cref{toeplitzrecovery}, and its theoretical guarantees as stated in Theorem \ref{sublinearqueryalgo}.

	We treat the $\wt{T}=F_{\wt{S}}\wt{D}F_{\wt{S}}^{*}$ that is guaranteed to exist by Theorem \ref{frobeniusexistence} as the true underlying matrix and $T$ as a noisy version of it that we have access to. Note that $T$ is completely determined by its first column $T_{1}\in \mathbb{R}^{d}$ and the first column of $\wt{T}$ is $F_{\wt{S}}\wt{d}$ where $\wt{d}\in \mathbb{R}^{\tilde O(k/\epsilon)}$ is the vector of values on the diagonal of $\wt{D}$.
	
	This suggests the following strategy - if the algorithm knew $\wt{S}$ exactly, it could try to find $\wt{d}$ by solving the  regression problem $\argmin_{a\in \mathbb{R}^d}\|F_{\wt{S}}a-T_1\|_2$. We can solve this regression problem approximately without reading too many entries of $T_1$ using the technique of leverage score sampling \cite{drineas2006sampling}, and using universal leverage score upper bounds for off-grid Fourier matrices that are independent of the set of frequencies $\wt{S}$ defining $F_{\wt{S}}$ \cite{chen2016fourier,eldar2020toeplitz}.
	
	The first issue with this strategy is the fact that the error bounds in fitting the first column $T_1$ would not translate to error bounds on fitting the entire matrix $T$ in the Frobenius norm. This is because each entry of $T_1$ appears a different number of times in $T$. For example $T_{1,1}$ appears $d$ times whereas $T_{d,1}$ only appears twice. To circumvent this, we need to solve an alternate regression problem where each row of $[F_{\wt{S}};T_1]$ is weighted differently to account for the asymmetry in fitting the first column $T_1$ versus fitting the entire matrix $T$. Using a geometric grouping technique we are able to obtain leverage score upper bounds for this modified regression problem that only suffer a logarithmic overhead compared to those known for the unweighted case \cite{eldar2020toeplitz}.
	
	The second issue is the assumption on the knowledge of $\wt{S}$ which can be circumvented by brute-force searching over the set of all possible $\wt{S}$ and choose the one with the smallest error.
	

	This suggests the following algorithm. It first obtains a sample set containing a few entries of $T_1$ using this universal leverage score distribution that is valid for any off-grid Fourier matrix, then searches for all possible sets of $\wt{O}(k)$ frequencies that could be the set $\wt{S}$. For each such guess of $\wt{S}$ it finds an approximately optimal $\wt{d}$ by approximately solving the weighted regression problem described previously, using the same sample set that works for any set $\wt{S}$. Finally, it returns the best $\wt{d}$ among all the guesses. 
	
	To prove that this algorithm works with good probability is still challenging. This is because standard sample-efficient regression results based on leverage score sampling \cite{Sarlos:2006vn,Woodruff:2014tg} do not suffice in our setting for two reasons. First, we search over many possible $\wt S$ and thus solve many regression problems -- we must take a union bound to argue that our sample set gives a good approximation for all these problems. This presents an issue for standard results, which typically require sample complexity depending linearly on $1/\eta$, where $1-\eta$ is the probability of success. Second, we require identifying a frequency set $\wt S$ with near minimal error -- i.e., we must compare the errors of the many regression problems that we solve. Standard leverage score based sampling results however, typically do not output an estimate of the actual regression error, making it impossible to chose a near optimal $\wt S$. To overcome these issues, we use a two stage algorithm, as in \cite{musco2021active}. Following techniques of \cite{eldar2020toeplitz}, we first find $\tilde S_1$ achieving a  constant factor of the optimal. We then find $\tilde S_2$ which gives a $(1+\epsilon)$ relative error fit to the residual remaining after regressing onto $\tilde S_1$. Our final frequency subset is $\tilde S_1 \cup \tilde S_2$. This approach allows us to use a modified analysis of leverage score sampling for fitting $\tilde S_2$, which both gives high probability bounds (with $\log(1/\eta)$ dependence for failure probability $\eta$) and regression error estimates, as required. A detailed description of the algorithm together with its analysis are presented in Algorithm \ref{algorithm}.
	
\section{Notation and preliminaries.} 

In this section, we introduce notation and preliminary concepts that are used throughout this paper.

\label{sec:preliminaries}
\subsection{Notation.}
For any functions $f,g:\mathbb{R}\rightarrow \mathbb{R}$, $f(n)\lsim g(n)$ means that there exists a constant $C>0$ such that  $f(n)\leq C g(n)$. For any positive integer $n$, let $[n] = \{1,2,\ldots,n\}$. For any set $S\subset \mathbb{R}$,  let $-S$ denote the set obtained by negating each element in $S$ and let $+S$ denote $S$ itself. For any set $N$, let $N^n$ denote the set of all subsets of $N$ with $n$ elements.
For a matrix $A$, let $A^{T}$ and $A^*$ denote its transpose and Hermitian transpose, respectively. For any vector $x\in \mathbb{C}^d$, let $\|x\|_2 = \sqrt{x^*x}$ denote its $\ell_2$ norm. For a matrix $A$ with $d$ columns, let $\|A\|_2 = \sup_{x\in \mathbb{C}^d} \|Ax\|_2/\|x\|_2$ denote its spectral norm and $\|A\|_F=\sqrt{\sum_{i \in [d]} \sum_{j \in [d]} A_{i, j}^2}$ denote its Frobenius norm. For a square matrix $A$, let $tr(A)$ denote its trace. 

A Hermitian matrix $A\in \mathbb{C}^{d\times d}$ is positive semidefinite (PSD) if for all $x\in \mathbb{C}^d$, $x^*Ax \geq 0$. Let $\lambda_{1}(A)\geq \ldots \geq \lambda_{d}(A)\geq 0$ denote its eigenvalues. Let $\preceq$ denote the Loewner ordering, that is $A\preceq B$ if and only if $B-A$ is PSD. Let $A= U\Sigma V^*$ denote the compact singular value decomposition of $A$, and when $A$ is PSD note that $U \Sigma U^*$ is its eigenvalue decomposition and let $A^{1/2} = U\Sigma^{1/2}$ denote its matrix square root, where $\Sigma^{1/2}$ is obtained by taking the elementwise square root of $\Sigma$. Let $A_k = U_k \Sigma_k V_k^{*}$ denote the projection of $A$ onto its top $k$ singular vectors. Here, $\Sigma_k\in \mathbb{R}^{k\times k}$ is the diagonal matrix containing the $k$ largest singular values of $A$, and $U_k, V_k \in \mathbb{C}^{d\times k}$ denote the corresponding $k$ left and right singular vectors of $A$. Note that $A_k$ is the optimal rank $k$ approximation to $A$ in the spectral and Frobenius norms, that is $A_{k} = \argmin_{\text{rank $k$ }\wt{A}}\|A-\wt{A}\|_2$ and $A_{k} = \argmin_{\text{rank $k$ }\wt{A}}\|A-\wt{A}\|_F$. Finally, for any vector $y\in \mathbb{R}^{d}$, let $T(y)\in \mathbb{R}^{d\times d}$ denote the symmetric Toeplitz matrix whose first column is $y$. 

\subsection{Fourier analytic and linear algebra tools.}
Let $T$ denote a $d\times d$ symmetric PSD Toeplitz matrix. We heavily rely on the Fourier structure of Toeplitz matrices.

While circulant matrices can be diagonalized by the discrete Fourier transform, this does not hold in general for Toeplitz matrices. However, a classical result called the Caratheodory-Fejer-Pisarenko decomposition (also called the Vandermonde Decomposition), which we also stated in Section \ref{sec:overview}, says that they still can be decomposed into a product of off-grid Fourier and diagonal matrices; this is formalized in the following result below.
\begin{reptheorem}{thm:vandermonde}
	Any real-valued PSD Toeplitz matrix $T \in \R^{d\times d}$ of rank $r$ can be expressed as $F_SDF_S^{*}$ where $F_S \in \mathbb{C}^{d\times r}$ is a symmetric Fourier matrix with frequencies $S\subseteq [0,1/2]$, $|S| = r/2$ and $D\in \mathbb{R}^{r\times r}$ is a diagonal matrix with $r$ positive entries. Moreover, for any $f\in S$ the values in $D$ corresponding to columns $v(f),v(-f)$, which we refer to as the weights of $f,-f$ respectively, are identical. Thus $D$ is uniquely defined by values $\{a_f\}_{f\in S}$.
\end{reptheorem}

We also define the wrap around distance between two freqencies as follows.
\begin{definition}[Wrap around distance]\label{wraparounddefn}
	The wrap around distance between any two frequencies $f,g\in [-1/2,1/2]$ is defined as $|f-g|_{\circ} = \min\{|f_i-f_j|,1-|f_i-f_j|\}$.
\end{definition}
Using the cyclic property of the trace, we can show the following Lemma, which relates the entries in $D$ to the eigenvalues of $T$.
\begin{lemma}\label{tracetrick}
	$2d\sum_{f\in S}a_f = tr(T) = \sum_{j=1}^d \lambda_j(T)$.
\end{lemma}
\begin{proof}
	Note that $tr(FDF^{*}) = tr(T) = tr(U\Sigma U^{T}) = \sum_{j=1}^d\lambda_j(T)$. Here $T=U\Sigma U^{T}$ is the eigendecomposition of $T$. Further, $tr(FDF^{*}) = tr(F^{*}FD)$. Since each diagonal entry of $F^{*}F$ is $d$ and $D$ is just a diagonal matrix with entries $\{a_{f}\}_{f\in S}$, we get that $tr(F^{*}FD)=2d\sum_{f\in S}a_f$. Thus $2d\sum_{f\in S}a_f = \sum_{j=1}^d \lambda_j(T)$.
\end{proof}

We will need the notion of \emph{statistical leverage scores} \cite{spielman2011graph,drineas2006sampling}, which are used to define non-uniform row sampling schemes. These schemes then enable randomized matrix compressions, which provide spectral approximation guarantees and preserve significant information \cite{cohen2015dimensionality,musco2017recursive}. Their precise definition is as follows.
\begin{definition}[Leverage score]\label{leveragescoresdefn}
	For any $A\in \mathbb{C}^{d\times r}$, let $\tau_j(A)$ denote the leverage score of the $j^{th}$ row of $A$:
	\begin{equation}
		\tau_j(A) = \max_{y\in \mathbb{C}^{s}} \frac{|(Ay)_i|^2}{\sum_{j=1}^d |(Ay)_j|^2}.
	\end{equation}
\end{definition} 

Finally, we will repeatedly use Weyl's eigenvalue perturbation bound for Hermitian matrices. 
\begin{theorem}[Weyl's inequality]\label{weylineq}
	For any $n>0$ and Hermitian matrices $B,C\in \mathbb{R}^{n\times n} $ and $A = B+C$, the following holds for all $i\in [n]$
	\begin{equation*}
		\lambda_{i}(B) - \|C\|_2 \leq \lambda_i(A) \leq \lambda_{i}(B)+\|C\|_2.
	\end{equation*}
\end{theorem}
For a proof of Weyl's inequality, we refer the reader to Section 1.3 in \cite{tao2012topics}.
\section{Existence of a near optimal low-rank approximation which itself is Toeplitz.}
\label{sec:existence_main}

The goal of this section is to prove Theorems \ref{spectralexistence} and \ref{frobeniusexistence}. 
Recall that we are given PSD Toeplitz matrix $T$ and $T=F_SDF_S^{*}$, where $F_S$ is a symmetric Fourier matrix (recall Definition \ref{def:symmetric_fourier_matrix}), with frequencies $S\subset [0,1/2]$ and  $D= diag(\{a_f\}_{f\in S})$ in its Vandermonde decomposition. We briefly discuss how this section is organized. In subsection \ref{case1}, we consider the case when all the frequencies in $S$ are very close to each other; this is the \emph{clustered} case.
Then in subsection \ref{case2}, we tackle the general case when the frequencies in $S$ are not necessarily clustered. We partition the frequency domain into buckets, where the weight of each bucket is the sum of the weights of all frequencies in $S$ landing in that bucket. The formal definition of buckets is as follows.
\begin{definition}\label{bucketing} 
	For any PSD Toeplitz matrix $T=F_SDF_S^{*}$, where $F_S$ is a Fourier matrix with frequencies $S\subset [0,1/2]$,  $D= diag(\{a_f\}_{f\in S})$, and $j$ is a positive integer, we define the $j$-th bucket $B_j$ by 
	$$
	B_j := \left[\frac{j-1}{d},\frac{j}{d}\right)\cap S,
	$$
	and let
	$$
	w(B_j)=\sum_{f\in B_j} a_f,
	$$
	denote the weight of the $j$-th bucket.
\end{definition}
The main claim in this subsection is to show that having heavy buckets implies that $T$ has many large eigenvalues. This claim is shown in two steps. First, we show the claim for Toeplitz matrices that have a well-separated spectrum as defined formally below.
\begin{definition}[$(\lambda,w)$-well separated]\label{wellseparatedcase}
	A PSD Toeplitz matrix $T=F_SDF_S^{*}$ is said to be $(\lambda,w)$-well separated if all non-empty buckets of $T$ have weight exactly $\lambda$, and for any two non-empty buckets, the minimum wrap-around distance between any two frequencies in these two buckets is at least $w$.
\end{definition}
Second, we show how to reduce from proving the claim for general Toeplitz matrices, to proving it for well separated Toeplitz matrices.
In subsection \ref{case3}, we use the same tools to complete the characterization relating buckets and eigenvalues. In particular, we show that if all buckets have weight smaller than $\lambda$, then all eigenvalues of $T$ are smaller than $O(d\lambda\log d)$. 
Finally, in subsection \ref{finalproof}, we show how to use the structural statements proven in subsections \ref{case2} and \ref{case3} to prove Theorems \ref{frobeniusexistence} and \ref{spectralexistence}.
\subsection{Case of clustered frequencies.}\label{case1}
We begin with the case when all the frequencies in $S$ are clustered, as formalized by Definition \ref{clustered_case-mod}. We first state some preliminaries. Define the $t^{th}$ entry of the $k^{th}$ column denoted by $T_k(t)$ for $t\in [0,d-1]$ as follows 
\begin{equation*}
	T_{k}(t) = \sum_{f\in S} a_fe^{2\pi i (f^*+r_f) (t-(k-1))} + \sum_{f\in S} a_fe^{-2\pi i (f^*+r_f) (t-(k-1))} \quad \forall k\in [d].
\end{equation*}
Then the following lemma, which follows trivially from expanding $T$'s Vandermonde decomposition $F_SDF_S^{*}$, says that the $k^{\text{th}}$ column of $T$ is the evaluation of $T_{k}(t)$ at $t=0,\ldots,d-1$.
\begin{lemma}
	For all $(m,k) \in [d]\times [d]$, $T_{m,k} = T_{k}(m-1)$. 
\end{lemma}
\begin{proof}
	Denote by $F_{S, i}$ the $i$-th row of $F_S$. Then
	\begin{multline*}
		T_{m, k} = (F_S D F_S^*)_{m, k} = F_{S, m} D F_{S, k}^* =\\
		\sum_{f\in S} a_f e^{2\pi i (f^*+r_f) m} e^{-2\pi i (f^*+r_f) k} + \sum_{f\in S} a_f e^{-2\pi i (f^*+r_f) m} e^{2\pi i (f^*+r_f) k} = \\
		\sum_{f\in S} a_f e^{2\pi i (f^*+r_f) (m - k)} + 			\sum_{f\in S} a_f e^{-2\pi i (f^*+r_f) (m - k)} = T_k(m - 1).
	\end{multline*}
	This concludes the proof of the lemma.
\end{proof}

Our main result in this section is stated in the following lemma.
\begin{replemma}{lem:clusteredcase}
	There exists a universal constant $C_1>0$ such that, given any symmetric PSD Toeplitz  $T$ that is $(f^*,\Delta)$-clustered for some $f^*\in [0,1/2]$ and $\Delta \leq 1/d$, and  for $0<\varepsilon,\delta,\gamma<1$ satisfying $\gamma\leq \varepsilon/(\tr(T)2^{C_1\log^7 d})$, the following conditions hold.  There exists a symmetric Toeplitz matrix $\wt{T}=F_{\wt{S}}\wt{D}F_{\wt{S}}^{*}$ of rank at most $O(\ell)$ for $\ell=O(\log d + \log(1/\delta))$ such that
	\begin{enumerate}
		\item $\wt{S}=\{f^*+j\gamma\}_{j=1}^{\ell+1}\cup \{f^*-j\gamma\}_{j=1}^{\ell+1}$, and $\wt{D}$ is a diagonal matrix such that for any $f\in \wt{S}$, the weights corresponding to $f$ and $-f$ are identical. 
		\item $\|T-\wt{T}\|_{F} \leq \delta (\sum_{f\in S}a_{f}) +\varepsilon d$.
	\end{enumerate} 
\end{replemma}
\begin{remark}
	The value of the diagonal elements in $\widetilde{D}$ depends on $\varepsilon$, and in particular the value of the diagonal elements would go to infinity as $\varepsilon$ goes to zero. However, we do not state this tradeoff explicitly in the lemma, since we use this lemma later in Section 4 to obtain leverage score upper bounds on Fourier sparse functions (based on the work of \cite{chen2016fourier}). These bounds do not depend on the magnitude of the coefficients in the function; instead, they depend only on its sparsity, as long as the coefficients are finite. Thus $\varepsilon$ can be set to any strictly positive number of one's choice.
\end{remark}

\begin{proof}

	Define for all $t\in [-d,d]$
	\begin{equation}\label{eq:t-def}
		\begin{split}
			T(t) &= \sum_{f\in S} a_ f e^{2\pi i (f^*+r_f) t} + \sum_{f\in S} a_fe^{-2\pi i (f^*+r_f) t}.
		\end{split}.
	\end{equation}

	Note that $T_{k}(t)$ is just the restriction of $T(t)$ to $t\in [-(k-1),d-k]$ for all $k\in [d]$. Thus we focus on uniformly approximating $T(t)$ over $t\in [-d,d]$, as this will yield uniform approximations for each $T_k(t)$ simultaneously.

	The next Lemma shows how to approximate $T(t)$ with a modulated low degree polynomial. The proof of this Lemma is deferred to the end of this subsection.
	\begin{lemma}\label{lowdeg_taylorapprox}
		For polynomials $p_1,p_2$ of degree $\ell=O(\log d +\log(1/\delta))$ defined as \\$p_{1}(t) = \sum_{m=0}^{\ell} \left[\sum_{f\in S} a_f\frac{(2\pi ir_f)^m}{m!}\right]t^m$ and $p_{2}(t) = \sum_{m=0}^{\ell}$ $ \left[\sum_{f\in S} a_f\frac{(-2\pi ir_f)^m}{m!}\right]t^m$, the following equality holds:
		\begin{equation*}
			|T(t) - e^{2\pi i f^* t}p_{1}(t) - e^{-2\pi i f^* t}p_{2}(t) |\leq \delta(\sum_{f\in S} a_f)  \quad \forall t\in [-d,d].
		\end{equation*}
	\end{lemma} 
	
	Observe that we can write the $p_{1}(t)$ and $p_{2}(t)$ obtained from Lemma~\ref{lowdeg_taylorapprox} as $p_1(t) = p_{even}(t)+ ip_{odd}(t)$ and $p_2(t) = p_{even}(t) - ip_{odd}(t)$, where $p_{even}$ and $p_{odd}$ contain the even and odd powered terms in $\sum_{m=0}^{\ell} \left[\sum_{f\in S} a_f\frac{(2\pi ir_f)^m}{m!}\right]t^m$ respectively. Now we can use the following lemma, which is a minor variation of Lemma 8.8 in \cite{chen2016fourier}, to express $p_{even}(t)$ and $p_{odd}(t)$ as $O(\log d+\log(1/\delta))$-Fourier sparse functions. For completeness, we present its proof at the end of this section.
	\begin{lemma}[Lemma 8.8 in \cite{chen2016fourier}]\label{polytofourier}
		Let $p(t)$ be a degree $\ell$ polynomial with coefficients $c_1,\ldots,c_{\ell}$, defined over $t\in [-d,d]$ and containing only even powers of $t$. For every  $\varepsilon\in (0, 1)$ and $\gamma \in (0, 1)$ satisfying
		\begin{equation*}
			\gamma \leq \varepsilon/(d2^{\Theta({\ell}^3\log (\ell))}\underset{1\leq i \leq p}{\max}|c_i|),
		\end{equation*}
		there exists $\wt{p}(t) = \sum_{j=1}^{\ell+1} \alpha_j(e^{2\pi i (j\gamma)t }+e^{-2\pi i (j\gamma)t })$ such that 
		\begin{equation*}
			|p(t) - \wt{p}(t)| \leq \varepsilon \quad \forall t\in [-d,d].
		\end{equation*}
		If $p(t)$ has only odd powers of $t$, then there exists $\wt{p}(t) = -i \sum_{j=1}^{\ell+1} \beta_j(e^{2\pi i (j\gamma)t }-e^{-2\pi i (j\gamma)t })$ that satisfies the above guarantee.
	\end{lemma}
	
	We apply Lemma~\ref{polytofourier} to approximate $p_{even}(t)$ by $\wt{p}_{even}(t)=\sum_{j=1}^{\ell+1}\alpha_j(e^{2\pi i (j\gamma)t}+e^{-2\pi i (j\gamma)t})$ and $p_{odd}(t)$ by $\wt{p}_{odd}(t)=-i\sum_{j=1}^{\ell+1}\beta_j(e^{2\pi i (j\gamma)t}-e^{-2\pi i (j\gamma)t})$. Note that since $r_f\in [0,1]$ for all $f\in S$, we can trivially upper bound the absolute value of the $j^{th}$ coefficient of $p_{even}$ and $p_{odd}$ by $d\sum_{f\in S}a_f$. Thus we can choose any $\gamma \leq \varepsilon/(d(\sum_{f\in S} a_f)2^{C_1\log^7 d})$ while applying Lemma~\ref{polytofourier}, since $\ell^3 \log (\ell)\leq O(\log^7 d)$. Now define $\wt{T}(t)$ as follows:
	\begin{align*}
		\wt{T}(t) &= e^{2\pi i f^*t}(\wt{p}_{even}(t)+i\wt{p}_{odd}(t)) +  e^{-2\pi i f^*t}(\wt{p}_{even}(t)-i\wt{p}_{odd}(t))\\
		&=\sum_{j=1}^{\ell+1}\alpha_j\left[(e^{2\pi i (f^*+j\gamma)t}+e^{-2\pi i (f^*+j\gamma)t})+(e^{2\pi i (f^*-j\gamma)t}+e^{-2\pi i (f^*-j\gamma)t})\right]\\
		&+\sum_{j=1}^{\ell+1}\beta_j\left[(e^{2\pi i (f^*+j\gamma)t}+e^{-2\pi i (f^*+j\gamma)t})-(e^{2\pi i (f^*-j\gamma)t}+e^{-2\pi i (f^*-j\gamma)t})\right]\\
		&=\sum_{j=1}^{\ell+1}(\alpha_j+\beta_j)\left[(e^{2\pi i (f^*+j\gamma)t}+e^{-2\pi i (f^*+j\gamma)t})\right]+(\alpha_j-\beta_j)\left[(e^{2\pi i (f^*-j\gamma)t}+e^{-2\pi i (f^*-j\gamma)t})\right].
	\end{align*}
	Let $\wt{T}_{k}(t)$ be the restriction of $\wt{T}(t)$ to $t\in [-(k-1),d-k]$. Thus applying the error guarantees of Lemma~\ref{lowdeg_taylorapprox} and \ref{polytofourier} and applying triangle inequality, we get that for all $k\in [d]$
	\begin{equation*}
		|T_k(t) - \wt{T}_k(t)| \leq \delta (\sum_{f\in S} a_f) + \varepsilon \quad \forall t \in [-(k-1),d-k].
	\end{equation*}
	
	Now $\wt{T}_k(t)$ for $k\in\{1,\ldots,d\}$ naturally defines the rank $4(\ell+1)$ symmetric Toeplitz matrix $\wt{T}=F_{\wt{S}}\wt{D}F_{\wt{S}}^{*}$, where $\wt{S}=\{f^*+j\gamma\}_{j=1}^{\ell+1}\cup \{f^*-j\gamma\}_{j=1}^{\ell+1}$ and $\wt{D}=diag(\{\alpha_j+\beta_j\}_{j=1}^{\ell+1}\cup\{\alpha_j-\beta_j\}_{j=1}^{\ell+1})$. The previous equation and this observation immediately that for $\wt{T}$,
	\begin{equation*}
		\|T-\wt{T}\|_F \leq \sqrt{\sum_{(i,j)\in [d]\times [d]} ( \delta (\sum_{f\in S} a_f) + \varepsilon)^2} = \delta d(\sum_{f\in S} a_f) + \varepsilon d.
	\end{equation*}
	We redefine $\delta$ as $\delta/d$. This concludes the proof of statement 2 in Lemma~\ref{lem:clusteredcase}.
	Lemma~\ref{tracetrick} implies that choosing $\gamma \leq \varepsilon/(d(\sum_{f\in S} a_f)2^{C_1\log^7 d}) = \varepsilon/(\tr(T)2^{C_1\log^7 d}) $ suffices for the first claim of the lemma to hold. 
\end{proof}
We now state the proof of Lemma~\ref{lowdeg_taylorapprox}.
\begin{proof}
	First, using a Taylor expansion we can write $T(t)$ as follows
	\begin{align*}
		T(t) &= \sum_{f\in S} a_fe^{2\pi i (f^*+r_f) t} + \sum_{f\in S} a_fe^{-2\pi i (f^*+r_f)t}\\
		& = e^{2\pi i f^* t}\left(\sum_{f\in S} a_f(\sum_{m=0}^{\infty} \frac{(2\pi i r_f t)^m}{m!})\right)+e^{-2\pi i f^* t}\left(\sum_{f\in S} a_f(\sum_{m=0}^{\infty} \frac{(-2\pi i r_f t)^m}{m!})\right)\\
		& = e^{2\pi i f^* t}\sum_{m=0}^{\ell} \left[\sum_{f\in S} a_f\frac{(2\pi ir_f)^m}{m!}\right]t^m+e^{-2\pi i f^* t}\sum_{m=0}^{\ell} \left[\sum_{f\in S} a_f\frac{(-2\pi ir_f)^m}{m!}\right]t^m\\
		&+ e^{2\pi i f^* t}\left(\sum_{f\in S} a_f(\sum_{m=\ell}^{\infty} \frac{(2\pi i r_f t)^m}{m!})\right)+ e^{-2\pi i f^* t}\left(\sum_{f\in S} a_f(\sum_{m=\ell}^{\infty} \frac{(-2\pi i r_f t)^m}{m!} )\right).
	\end{align*}
	Using the fact that $|r_f|\leq \Delta$ and $t\in [-d,d]$, we have that $|2\pi i r_f|< 10 \Delta d$ for all $f\in S$. Also note that for any $m>1, m!\geq (m/2)^{m/2}$. This implies that the following holds for all $r_f$ and $t\in[-d,d]$
	\begin{align*}
		\left|\sum_{m=\ell}^{\infty} \frac{(2\pi i r_f t)^m}{m!}\right|&\leq \sum_{m=\ell}^{\infty} \left|\frac{(2\pi i r_f t)^m}{m!}\right|  \leq \sum_{m=\ell}^{\infty} (10\Delta d)^m/(m/2)^{m/2} \\
		&\leq \sum_{m=\ell}^{\infty} (10\Delta d/\sqrt{m/2})^m \leq \sum_{m=\ell}^{\infty} (10/100 \sqrt{\log d})^{m}\\
		&= \sum_{m=\ell}^{\infty} (0.1)^m \leq (0.1)^{O(\log(1/\delta))}/0.9 \leq \delta/2.
	\end{align*}
	Thus letting $p_{1}(t) = \sum_{m=0}^{\ell} \left[\sum_{f\in S} a_f\frac{(2\pi ir_f)^m}{m!}\right]t^m$ and $p_{2}(t) = \sum_{m=0}^{\ell}$ $ \left[\sum_{f\in S} a_j\frac{(-2\pi ir_f)^m}{m!}\right]t^m$,  we get that
	\begin{align*}
		|T(t) - e^{2\pi i f^*t}p_{1}(t) - e^{-2\pi i f^*t}p_{2}(t)| &\leq \sum_{f\in S} a_{j}\left(\left|\sum_{m=\ell}^{\infty} \frac{(2\pi i r_f t)^m}{m!} \right|+ \left|\sum_{m=\ell}^{\infty} \frac{(-2\pi i r_f t)^m}{m!} \right|\right)\\&\leq \delta(\sum_{f\in S} a_j),
	\end{align*}
	where in the second inequality we used that all the $a_j$'s are non-negative.
\end{proof}
We finish this subsection with the proof of Lemma~\ref{polytofourier}.\begin{proof}
	The proof of this is almost identical to the proof of Lemma 8.8 in \cite{chen2016fourier}. We first consider the case when the degree $\ell$ polynomial $p(t)$ has only even powers of $t$. Let $p(t) = \sum_{j=0}^{\ell/2} c_{2j}t^{2j}$. First, note we can write $\wt{p}(t)$ as 
	\begin{align*}
		\wt{p}(t) &= \sum_{j=1}^{\ell+1} \alpha_{j}( e^{2\pi i (\gamma j)t}+e^{-2\pi i (\gamma j)t})\\
		&= \sum_{j=1}^{\ell+1} 2\alpha_{j} (\sum_{k=0:k \text{ even}}^{\infty} \frac{(2\pi i \gamma j t)^k}{k!})\\
		&= \sum_{k=0:k\text{ even}}^{\infty} 2\frac{(2\pi i \gamma  t)^k}{j!}\sum_{j=1}^{\ell+1} \alpha_j j^k\\
		&= \sum_{k=0:k\text{ even}}^{\ell} 2\frac{(2\pi i \gamma  t)^k}{j!}\sum_{j=1}^{\ell+1} \alpha_j j^k+\sum_{k=\ell+1:k\text{ even}}^{\infty} 2\frac{(2\pi i \gamma  t)^k}{j!}\sum_{j=1}^{\ell+1} \alpha_j j^k\\
		& = p(t)+ (\sum_{k=0:k\text{ even}}^{\ell} 2\frac{(2\pi i \gamma  t)^k}{j!}\sum_{j=1}^{\ell+1} \alpha_j j^k-p(t))+\sum_{k=\ell+1:k\text{ even}}^{\infty} 2\frac{(2\pi i \gamma  t)^k}{j!}\sum_{j=1}^{\ell+1} \alpha_j j^k.
	\end{align*}
	We will show that there exists some $\gamma$ and $\alpha_1,\ldots,\alpha_{\ell+1}$ such that $\sum_{k=0:k\text{ even}}^{\ell} 2\frac{(2\pi i \gamma  t)^k}{j!}\sum_{j=1}^{\ell+1} \alpha_j j^k-p(t)$ is identically zero. The $\ell+1 \times \ell+1$ Vandermonde matrix $A'$ is the same as in the proof of Lemma 8.8 in \cite{chen2016fourier}; however, the $\ell+1$ dimensional vector $c'$ is slightly different: $c'_{j} =  c_j j!/(2\pi i \gamma)^j$ for $j$ even and $c'_j=0$ for $j$ odd. The rest of the argument is identical to the proof of Lemma 8.8 in \cite{chen2016fourier}.
	Now, the argument to show that the absolute value of $C_2 = \sum_{k=\ell+1:k\text{ even}}^{\infty} 2\frac{(2\pi i \gamma  t)^k}{j!}\sum_{j=1}^{\ell+1} \alpha_j j^k$ is less than $\epsilon$ for all $|t|\leq d$ (for the choice of $\gamma$ and $\alpha_1,\ldots,\alpha_{\ell+1}$ shown to exist previously) is also nearly identical to the proof of Lemma 8.8 in \cite{chen2016fourier}. The only difference is that, while upper bounding the value of $|C_2|$, we replace $t^j$ with $|t|^j$ for all $j\in [\ell+1,\infty]$.
	
	The proof for the case when $p(t)$ has only odd powered terms is identical to the proof discussed above, with the exception that the goal is to show the existence of  $\wt{p}(t) = -i\sum_{j=1}^{\ell+1} \beta_j(e^{2\pi i(\gamma j)t}-e^{-2\pi i(\gamma j)t})$ satisfying the claim of the theorem. 
\end{proof}

\subsection{Many heavy buckets implies many large eigenvalues.}\label{case2}

In this subsection, we present our first structural result relating the Vandermonde decomposition $T = F_S D F_{S}^{*}$ to the eigenvalue decomposition of $T=U\Sigma U^{T}$. We first bucket the frequencies in the Vandermonde decomposition, as described in Definition \ref{bucketing}. 

Equipped with this definition, we now state the main result of this section.
\begin{replemma}{heavybucketlemma}
	For every $\lambda,k>0$, if $T$ has at least $k$ buckets with weight at least $\lambda$, then $T$ has at least $k/\log^3 d$ eigenvalues that have value $\Omega(d\lambda/\log d)$. 
\end{replemma}
Before we prove Lemma~\ref{heavybucketlemma}, we state the following helper lemmas. We will then present the proof of Lemma~\ref{heavybucketlemma} using these helper lemmas, and finally end this subsection with their proofs. 
\begin{lemma}\label{innerproduct}
	For any two frequencies $f,g\in [-1/2,1/2]$, 
	\begin{align*}
		|v(f)^{*}v(g)|\leq O(1/|f-g|_{\circ}),
	\end{align*}
	where $v(f),v(g)$ are the corresponding frequency vectors (recall Definition \ref{freqvector} of frequency vectors) and $|f-g|_{\circ}$ is the wrap-around distance between $f$ and $g$.  
\end{lemma}

The next lemma essentially is a strengthening of Gershgorin's circle theorem for block matrices. 
\begin{lemma}[Section 1.13 in \cite{tretter2008spectral}, or Corollary 3.2 in \cite{echeverria2018block}]
	\label{blockgershgorin}
	Let $A = (A_{ij})\in \mathbb{C}^{dn\times dn}$ be a Hermitian matrix composed of blocks $A_{ij}\in \mathbb{C}^{d\times d}$. Then each $A_{ii}$ is Hermitian, and the following is true:
	\begin{equation*}
		\|A\|_2\leq \max_{i\in [n]}(\|A_{ii}\|_2+\sum_{\substack{j\in [n]\\ j\neq i}}\|A_{ij}\|_2).
	\end{equation*}
\end{lemma}
The final helper lemma will help us upper bound the norms of the off-diagonal blocks when applying Lemma~\ref{blockgershgorin}. 
\begin{lemma}\label{frobnormupperbound}
	For any $\lambda\geq 0$, $\sigma_1,\sigma_2\in \{+,-\}$, sets of frequencies $S_1,S_2\subseteq [0,1/2]$ and corresponding diagonal weight matrices $D_1,D_2$ both of whose traces are at most $\lambda$, $A = D_1^{1/2}F_{\sigma_1S_1}^{*}F_{\sigma_2S_2}D_{2}^{1/2}$ satisfies $\|A\|_F \leq O(\lambda/d(\sigma_1S_1,\sigma_2S_2))$. Here recall that for any set of frequencies $S\subset [0,1/2]$, $+S$ and $-S$ contain the frequencies in $S$ and their negations, respectively, and $d(A,B)$ for any sets $A,B\subseteq[-1/2,1/2]$ is the minimum wraparound distance between any two frequencies in $A,B$.
\end{lemma}

We first consider a Toeplitz matrix with a well separated spectrum (as per Definition \ref{wellseparatedcase}) and prove Lemma~\ref{heavybucketlemma} for this case. Then, we reduce the general case to the well separated case, which proves Lemma~\ref{heavybucketlemma} in full generality. 

It can be seen from the following lemma that, essentially, we can reduce from a general Toeplitz matrix to the well-separated case.
\begin{lemma}\label{wellseparatedreduction}
	For any $w,\lambda,k>0$ and a PSD Toeplitz matrix $T$ having at least $k$ buckets with weight at least $\lambda$, there exists a $(\lambda,w)$-well separated PSD Toeplitz matrix $T_{separated}$ having at least $k/dw$ non-empty buckets satisfying $T_{separated}\preceq T$.
\end{lemma}

The following lemma essentially proves Lemma~\ref{heavybucketlemma} for the special case of well-separated matrices. 
\begin{lemma}\label{heavybucketspecialcase}
	For every $\lambda,w,k>0$ and a $(\lambda,w)$-well separated $T$ having at least $k$ non-empty buckets, $T$ has at least $k$ eigenvalues that have value at least $\Omega(\lambda(d/\log d - \log d/w))$.
\end{lemma}

Equipped with these helper lemmas, we are ready to present the proof of Lemma~\ref{heavybucketlemma}.
\begin{proof}
	Refer to the $k$ buckets with weight at least $\lambda$ as \emph{heavy buckets}. Without loss of generality, we can assume that each bucket has weight exactly $\lambda$. This can be seen by the following argument. We reduce the weight of each heavy bucket if needed to ensure that each heavy bucket has weight exactly $\lambda$. This change will only make $T$ PSD smaller, and thus it suffices to lower bound eigenvalues of $T$ after this operation. We can easily obtain a $(\lambda,w)$-well separated Toeplitz matrix from $T$ as per Claim \ref{wellseparatedreduction}.
	
	Consider the matrix $T_{separated}$ guaranteed to exist by Lemma \ref{wellseparatedreduction}. Thus, invoking Lemma~\ref{heavybucketspecialcase} for $T_{separated}$ with parameters $\lambda,k/dw,w$ for $w=\log^3d/d$, we get that $T_{separated}$ has at least $k/\log^3d$ eigenvalues of value $\Omega(d\lambda (1/\log d - 1/\log^2d)) = \Omega(d\lambda/\log d)$. Since $T_{separated} \preceq T$, $T$ also has at least $k/\log^3d$ eigenvalues of value $\Omega(d\lambda (1/\log d - 1/\log^2d)) = \Omega(d\lambda/\log d)$. This finishes the proof of Lemma~\ref{heavybucketlemma}.
\end{proof}
Now we present the proof of the well separated case (Lemma~\ref{heavybucketspecialcase}).
\begin{proof}
	Let $T=F_S D F_S^{*}$ denote the Vandermonde decomposition of $T$.
	Observe that $F_S D F_S^{*}$ has the same eigenvalues as $D^{1/2}F_S^{*}F_SD^{1/2}$. This is because, for any matrix $A$, $A^* A$ and $AA^{*}$ have the same eigenvalues; apply this fact to $A=F_SD^{1/2}$. It can be easily seen that $D^{1/2}F_S^{*}F_SD^{1/2}$ is a block matrix, where each of its blocks is of the form $D_{m}^{1/2}F_{B_m}^{*}F_{B_n}D_{n}^{1/2}$ for some $m,n$, and $D_m$ is the diagonal matrix containing weights of the frequencies in bucket $B_m$ for every $m$. We can pad each block with enough zero rows and columns to ensure that the dimensions of each block are the same, without affecting its spectral norm. Thus we can express $D^{1/2}F_S^{*}F_SD^{1/2} = A+E$, where $A$ contains the diagonal blocks with zero matrices on the off diagonal blocks, and $E$ contains the off-diagonal blocks with zero matrices on the diagonal blocks. Our high level strategy is to prove the theorem statement for $A$ by using the fact that it is block diagonal. Then, we show that $E$ has small spectral norm, so the eigenvalues of $A+E$ are close to the eigenvalues of $A$.\\
	\\
	\textbf{Lower bounding largest eigenvalue of diagonal blocks.} $A$ is block diagonal with at least $k$ blocks. Consider any of its blocks $D_{m}^{1/2}F_{B_m}^{*}F_{B_m}D_{m}^{1/2}$ for some $m$. We know that $F_{B_m}D_{m}F_{B_m}^{*}$ is $((m-1/2)/d,1/d)$-clustered. Thus by Lemma~\ref{lem:clusteredcase} with $\delta=1/2d^{11}$ and $\varepsilon = \lambda/2d^{12}$,$F_{B_m}D_{m}F_{B_m}^{*}$ is  $\delta(\sum_{f\in B_m}a_{f}) + \varepsilon d = \lambda/d^{11}$ close to a symmetric Toeplitz matrix $\wt{T}$ with rank at most $C\log d$ in the Frobenius norm for some universal constant $C>0$. Thus by Weyl's inequality (Theorem \ref{weylineq}), we get the following for any $k>C\log d$
	
	\begin{align*}
		\lambda_{k}(F_{B_m}D_{m}F_{B_m}^{*})&\leq \lambda_{k}(\wt{T})+\|F_{B_m}D_{m}F_{B_m}^{*}-\wt{T}\|_{2}\\
		&\leq 0+ \|F_{B_m}D_{m}F_{B_m}^{*}-\wt{T}\|_{F}\\
		&\leq \lambda/d^{11}.
	\end{align*}
	Thus we get the following bound on the trace of $F_{B_m}D_{m}F_{B_m}^{*}$
	\begin{equation*}
		\tr(F_{B_m}D_{m}F_{B_m}^{*}) \leq C\log d\lambda_{1}(F_{B_m}D_{m}F_{B_m}^{*})+\lambda/d^{10}.    
	\end{equation*}
	This implies that
	\begin{align*}
		\lambda_{max}(D_{m}^{1/2}F_{B_m}^{*}F_{B_m}D_{m}^{1/2})&=\lambda_{1}(F_{B_m}D_{m}F_{B_m}^{*})\\
		&\geq \tr(F_{B_m}D_{m}F_{B_m}^{*})/(C\log d) - \lambda/(Cd^{10}\log d)\\
		&=d\lambda/(C\log d) - \lambda/(Cd^{10}\log d) = \Omega(d\lambda/\log d).
	\end{align*}
	Since $A$ is a block diagonal matrix with at least $k$ blocks, and each block has largest eigenvalue $\Omega(d\lambda/\log d)$, we get that $A$ has at least $k$ eigenvalues which are at least $\Omega(d\lambda/\log d)$ in value.\\
	\\
	\textbf{Upper bounding the contribution of off-diagonal blocks.} In the rest of the proof, we bound the spectral norm of $E$. We do this by applying Lemma~\ref{blockgershgorin} to $E$ where its diagonal blocks are $0$ and any off diagonal block is of the form $D_{m}^{1/2}F_{B_m}^{*}F_{B_n}D_{n}^{1/2}$ for some buckets $B_m,B_n$ satisfying $m\neq n$. 
	Consider an arbitrary but fixed bucket index $m$. In order to apply Lemma~\ref{blockgershgorin}, we need to upper bound $\sum_{n:n\neq m} \|D_{m}^{1/2}F_{B_m}^{*}F_{B_n}D_{n}^{1/2}\|_2\leq \sum_{n:n\neq m} \|D_{m}^{1/2}F_{B_m}^{*}F_{B_n}D_{n}^{1/2}\|_F$. We recall that for any $n$, $F_{B_n}= [F_{+B_n};F_{-B_n}]$ as per the notation in Definition \ref{def:symmetric_fourier_matrix}. We apply Lemma~\ref{frobnormupperbound} to bound $\|D_{m}^{1/2}F_{B_m}^{*}F_{B_n}D_{n}^{1/2}\|_F$ as follows
	\begin{align*}
		\|D_{m}^{1/2}F_{B_m}^{*}F_{B_n}D_{n}^{1/2}\|_F &\leq \|D_{m}^{1/2}F_{+B_m}^{*}F_{+B_n}D_{n}^{1/2}\|_F+\|D_{m}^{1/2}F_{+B_m}^{*}F_{-B_n}D_{n}^{1/2}\|_F\\&+\|D_{m}^{1/2}F_{-B_m}^{*}F_{+B_n}D_{n}^{1/2}\|_F+\|D_{m}^{1/2}F_{-B_m}^{*}F_{-B_n}D_{n}^{1/2}\|_F\\
		&\leq O(\lambda(1/d(+B_m,+B_n)+1/d(+B_m,-B_n)+\\
		&1/d(-B_m,+B_n)+1/d(-B_m,-B_n))).
	\end{align*}
	Thus, we need to upper bound the following
	\begin{align*}
		\sum_{n:n\neq m}O(\lambda(1/d(+B_m,+B_n)+1/d(+B_m,-B_n)+
		1/d(-B_m,+B_n)+1/d(-B_m,-B_n))).
	\end{align*}
	It is easy to see that the maximum value the above expression can take will correspond to the case when, informally, there are as many buckets as possible in $S$,  while still ensuring that any two buckets have separation at least $w$. More formally, this corresponds to the case when for every $n\in [m-d/(2w),m+d/(2w)]$, there is a bucket at distance $d_{m,n} = \Theta(|m-n|w)$ from $B_m$. Thus, the above expression is bounded by
	\begin{align*}
		\sum_{n=m+1}^{m+d/(2w)}O(\lambda/((n-m)w))  &+\sum_{n=m-d/(2w)}^{m-1}O(\lambda/((m-n)w)) \\
		& \leq 2\sum_{l=1}^d O(\lambda/(lw))\\&= O(\lambda/w)\sum_{l=1}^{d}1/l\\
		&= O(\lambda \log d/w).
	\end{align*}
	Since any diagonal block has spectral norm $0$ and the previous bound holds for all $m$, Lemma~\ref{blockgershgorin} applied to $E$ implies that $\|E\|_2\leq O(\lambda \log d/w)$.
	Now, applying Weyl's inequality for all $i$ such that $\lambda_i(A)=\Omega( d\lambda/\log d)$, we have
	\begin{equation*}
		\lambda_i(A+E)\geq\lambda_{i}(A) - \|E\|_2 = \Omega(d\lambda/\log d) - \|E\|_2\geq \Omega(\lambda(d/\log d-\log d/w)).    
	\end{equation*}
	Since there are at least $k$ such indices $i$, this completes the proof of the Lemma.
\end{proof}
We now present the proof of Lemma~\ref{innerproduct}.
\begin{proof}
	Denote by $\delta = |f - g|_{\circ}$ the wrap around distance between $f$ and $g$. 
	Since $|v_{f}^{*}v_{g}| = |v_{f}^{*}v_{g}|$, we have the following. 
	
	\[
	|v(g)^{*}v(f)| = \left|\sum_{l = 0}^{d - 1} \exp(2 \pi i \delta l)\right| = \left|\frac{\exp(2 \pi i \delta d) - 1}{\exp(2 \pi i \delta) - 1}\right| = \left|\frac{\exp(\pi i \delta d)}{\exp(\pi i \delta)} \cdot \frac{\sin(\pi \delta d)}{\sin(\pi \delta)}\right|.
	\]
	Therefore, it holds that
	\begin{equation*}
		|v(g)^{*}v(f)| = \left| \frac{\sin(\pi \delta d)}{\sin(\pi \delta )} \right| \leq \frac{1}{|\sin(\pi \delta)|} \leq  O\left(\frac{1}{\delta}\right).
	\end{equation*}
	Here, the last inequality follows by the fact that $\sin(\pi x)\geq 2x$ for $x\in [0,1/2]$, and the wrap around distance between any two frequencies lies between $0$ and $1/2$.
\end{proof}
Now, we present the proof of Lemma~\ref{frobnormupperbound}.

\begin{proof}
	We have the following for all $i,j$:
	\begin{equation*}
		|(A)_{i,j}| = \sqrt{(D_1)_i(D_2)_j}|v(f_i)^{*}v(f_j)|\leq O(\sqrt{(D_1)_i(D_2)_j}/|f_i-f_j|)\leq O(\sqrt{(D_1)_i(D_2)_j}/d(\sigma_1S_1,\sigma_2S_2)),
	\end{equation*}
	where $(D_1)_i,(D_2)_j$ and $f_{i},f_{j}$ are the weights and frequencies of the $i^{th}$ frequency in the $\sigma_1S_1$ and $j^{th}$ frequency in the $\sigma_2S_2$ respectively. Also, note that the second to last inequality follows from Lemma~\ref{innerproduct}.
	This implies
	\begin{align*}
		\|A\|_{F}^2 &\leq O( \sum_{i,j}(D_1)_i(D_2)_j/d(\sigma_1S_1,\sigma_2S_2)^2) \\
		&=O((\sum_{i}(D_1)_i)(\sum_{j}(D_2)_j)/d(\sigma_1S_1,\sigma_2S_2)^2)\\
		&\leq O(\lambda^2/d(\sigma_1S_1,\sigma_2S_2)^2).
	\end{align*}
	Thus $\|A\|_F \leq O( \lambda/d(\sigma_1S_1,\sigma_2S_2))$.  
\end{proof}
We finally end this sub-section with the proof of Claim \ref{wellseparatedreduction}.
\begin{proof}
	Sort the $k$ heavy buckets of $T$ according to their central frequency, and index them by $[k]$. Now consider the Toeplitz matrix $T_{separated}$, obtained by taking the buckets indexed by the residue class $1 \mod dw$ out of the $k$ heavy buckets of $T$. There are at least $k/dw$ such buckets, and the pairwise separation between any two such buckets is at least $w$. This implies $T_{separated}$ is $(\lambda,w)$-well separated and has $k/dw$ non-empty buckets. Moreover, since $T_{separated}$ is obtained by taking a subset of the non-empty buckets of $T$, this implies that $T_{separated}\preceq T$.
\end{proof}
\subsection{All buckets being light implies all eigenvalues are small.}\label{case3}
In this subsection, we show that if all buckets of $T$ have small weight, then all eigenvalues of $T$ are small. This result, combined with the previous section, will suffice to prove Theorems \ref{spectralexistence} and \ref{frobeniusexistence}. This is formalized in the lemma below.

\begin{replemma}{lightbucketlemma}
	For every $\lambda\geq 0$, if all buckets of $T$ have weight at most $\lambda$, then $\|T\|_2\leq O(d\lambda\log d)$.
\end{replemma}
\begin{proof}
	We need to upper bound $\|F_SDF_S^{*}\|_2=\|D^{1/2}F_S^{*}F_SD^{1/2}\|_2$. Again, $A:=D^{1/2}F_S^{*}F_SD^{1/2}$ is a block matrix, where its blocks are of the form $A_{m,n}:=D_{m}^{1/2}F_{B_m}^{*}F_{B_n}D_{n}^{1/2}$ for some buckets $B_m,B_n$. Here $F_{B_m}$ and $D_{m}$ are the Fourier and diagonal weight matrices, respectively, of the $m^{th}$ bucket. To prove the Lemma, we will apply Lemma~\ref{blockgershgorin} on $A$. To do so, we need to upper bound $\|A_{m,m}\|_2+\sum_{n:n\neq m}\|A_{m,n}\|_2\leq \|A_{m,m}\|_2+\sum_{n:n\neq m}\|A_{m,n}\|_F$ for all $m$. Since $A_{m,m}$ is PSD, we upper bound $\|A_{m,m}\|_2$ as follows
	\begin{equation*}
		\|A_{m,m}\|_2\leq \tr(A_{m,m}) = d\sum_{i} (D_m)_i =dw(B_m)\leq d\lambda.
	\end{equation*}
	We can trivially upper bound $\|A_{m,m+1}\|_F$ by first applying Cauchy-Schwarz to each entry $(A_{m,m+1})_{i,j}$:
	\begin{align*}
		|(A_{m,m+1})_{i,j}| = \sqrt{(D_{m})_i(D_{(m+1)})_j}|v(f_{i})^{*}v(f_{j})|\leq & \sqrt{(D_{m})_i(D_{(m+1)})_j}\|v(f_{i})\|_2\|v(f_{j})\|_2\\ &= \sqrt{(D_{m})_i(D_{(m+1)})_j}d,
	\end{align*}
	where $f_{i},(D_m)_i$ are the $i^{th}$ frequency and its corresponding weight in the $m^{th}$ bucket respectively.
	Then we use the previous bound to get that the following holds for any $m$:
	\begin{equation*}
		\|A_{m,m+1}\|_F \leq \sqrt{(\sum_{i}(D_{m})_i)(\sum_{j}(D_{m})_j)d^2}=d\lambda.
	\end{equation*}
	The same bound holds for $\|A_{m,m-1}\|_F$. Thus we obtain the bound
	\begin{equation*}
		\|A_{m,m}\|_2+\sum_{n:n\neq m}\|A_{m,n}\|_F \leq 3d\lambda + \sum_{n:n\notin \{m-1,m,m+1\}} \|A_{m,n}\|_F.
	\end{equation*}
	Since $d(+B_m,+B_{m+1})$ and $d(+B_m,+B_{m-1})$ could be arbitrarily close to $0$ (adjacent buckets could have frequencies very close to each other), the bound $\|A_{m,m+1}\|_F\leq d\lambda$ is tight. For the remaining $n\notin \{m,m-1,m+1\}$, we use Lemma~\ref{frobnormupperbound} to bound $\|A_{m,n}\|_F$ as follows:
	\begin{align*}
		\|A_{m,n}\|_F&\leq O(\lambda(1/d(+B_m,+B_n)+1/d(+B_m,-B_n)+\\
		&1/d(-B_m,+B_n)+1/d(-B_m,-B_n))).
	\end{align*}
	Similar to the proof in the previous subsection, it is easy to see that the maximum value \\$\sum_{n:n\notin \{m-1,m,m+1\}} \|A_{m,n}\|_F$ can take will correspond to the case when all buckets in $[0,1/2]$ are non-empty. More formally, for every $n\in [m-d/2,m+d/2]\setminus \{m-1,m,m+1\}$, there is a bucket at distance $ \Theta(|m-n|/d)$ from bucket $B_m$. Thus we get that the above is bounded by
	\begin{align*}
		\sum_{n=m+2}^{m+d/2}O(\lambda/((n-m)/d))  +&\sum_{n=m-d/2}^{m-2}O(\lambda/((m-n)/d)) \\
		& \leq 2\sum_{l=1}^d O(\lambda/(l/d))\\&= O(d\lambda)\sum_{l=1}^{d}1/l\\
		&= O(d\lambda\log d).
	\end{align*}
	Thus applying Lemma~\ref{blockgershgorin} to $A$, we get that $\|T\|_2 = \|A\|_2 \leq 3d\lambda+ O(d \lambda\log d) = O(d\lambda\log d)$. This completes the proof of Lemma~\ref{lightbucketlemma}.
\end{proof}
\subsection{Existence of Toeplitz low-rank approximation in the spectral and Frobenius norm.}\label{finalproof}
In this section, we use Lemma~\ref{heavybucketlemma} and Lemma~\ref{lightbucketlemma} to show the existence of a near optimal low-rank approximation in the both the spectral and Frobenius norm, which itself is Toeplitz. The first claim is formalized in the following theorem.
\begin{theorem}\label{spectralexistence}
	For any PSD Toeplitz $T\in \mathbb{R}^{d\times d}$,  $0<\delta<1$ and any integer $k\leq d$,  there exists a symmetric Toeplitz matrix $\wt{T}$ of rank $\wt{O}(k\log(1/\delta))$ such that the following holds
	\begin{equation*}
		\|T-\wt{T}\|_2 \leq \wt{O}(1) \|T-T_k\|_2 + \delta\|T\|_F.
	\end{equation*}
	where $T_k=\underset{B:rank(B)\leq k}{\|T-B\|_2}$ is the best rank-$k$ approximation to $T$ in the spectral norm.
\end{theorem}
\begin{proof}
	Let $T=F_S D F_S$ be $T$'s Vandermonde decomposition and let $\lambda = \lambda_{k+1}(T)\log d/c'd$, where $c'$ is the same constant appearing in the big-O notation in Lemma~\ref{heavybucketlemma}. Bucket the frequencies in $S$ as per Definition \ref{bucketing}. Let $T_{heavy}$ be PSD Toeplitz matrix obtained by considering all buckets with weight strictly more than $\lambda$. Then $T-T_{heavy}$ is also a PSD Toeplitz matrix, which contains only buckets of $T$ that have weight bounded by $\lambda$. Thus by Lemma~\ref{lightbucketlemma} we have the following
	\begin{equation*}
		\|T-T_{heavy}\|_2\leq  O(d\lambda\log d) = \wt{O}(1)\|T-T_k\|_2.
	\end{equation*}
	Now we claim that $T$ has at most $(k+2)\log^3 d$ buckets with weight strictly more than $\lambda_{k+1}(T)\log d/c'd$. (If not, then by Lemma~\ref{heavybucketlemma}, $T$ has at least $k+2$ eigenvalues with value strictly more than $\lambda_{k+1}(T)$; this contradicts the definition of $\lambda_{k+1}(T)$ being the $(k+1)^{th}$ largest eigenvalue of $T$.) Thus, $T_{heavy}$ contains at most $(k+2)\log^3 d $ non-empty buckets. Let $r\leq (k+2)\log^3 d$ be the number of these heavy buckets, and define $F_{S_1},\ldots, F_{S_r}$ and $D_1,D_2,\ldots, D_r$ to be the Fourier and diagonal weight matrices corresponding to these $r$ heavy buckets which define $T_{heavy}$. We can therefore write $T_{heavy}$ as 
	\begin{equation*}
		T_{heavy} = \sum_{m=1}^{r} F_{S_m}D_{m}F_{S_m}^{*}.
	\end{equation*}
	Since each $ F_{S_m}D_{m}F_{S_m}^{*}$ is $(f,1/d)$-clustered for some $f\in [0,1/2]$, we use Theorem \ref{lem:clusteredcase} to approximate each $ F_{S_m}D_{m}F_{S_m}^{*}$ with a rank $O(\log d+\log(1/\delta))$ symmetric Toeplitz matrix $ F_{\wt{S}_m}\wt{D}_{m}F_{\wt{S}_m}^{*}$ that satisfies 
	\begin{equation*}
		\|F_{S_m}D_{m}F_{S_m}^{*} - F_{\wt{S}_m}\wt{D}_{m}F_{\wt{S}_m}^{*}\|_{F}\leq \delta(\sum_{i} (D_m)_{i}) + \varepsilon d.
	\end{equation*}
	This implies that
	\begin{align*}
		\left\|T-\sum_{m=1}^{r} F_{\wt{S}_m}\wt{D}_{m}F_{\wt{S}_m}^{*}\right\|_2 &\leq \left\|T-\sum_{m=1}^{r} F_{S_m}D_{m}F_{S_m}^{*}\right\|_2 + \left\|\sum_{m=1}^{r} (F_{S_m}D_{m}F_{S_m}^{*} - F_{\wt{S}_m}\wt{D}_{m}F_{\wt{S}_m}^{*}) \right\|_2\\
		& = \|T-T_{heavy}\|_2 + \left\|\sum_{m=1}^{r} (F_{S_m}D_{m}F_{S_m}^{*} - F_{\wt{S}_m}\wt{D}_{m}F_{\wt{S}_m}^{*} )\right\|_2\\
		&\leq \|T-T_{heavy}\|_2 + \sum_{m=1}^{r} \left\|F_{S_m}D_{m}F_{S_m}^{*} - F_{\wt{S}_m}\wt{D}_{m}F_{\wt{S}_m}^{*}\right\|_F\\
		& \leq \wt{O}(1) \|T-T_k\|_2 + \delta(\sum_{m,i}(D_m)_{i}) +\varepsilon r d \\
		&\leq  \wt{O}(1) \|T-T_k\|_2 + \delta(\sum_{i}(D)_{i}) + \varepsilon r d.
	\end{align*}
	Now using Lemma~\ref{tracetrick}, we get the following:
	\begin{equation*}
		\sum_{i}(D)_{i} = \sum_{i=1}^{d}\lambda_i(T)  /d \leq \left(\sqrt{d \sum_{i=1}^d \lambda_{i}(T)^2}\right)/d = \|T\|_F/\sqrt{d}.
	\end{equation*}
	Setting $\varepsilon = \delta\|T\|_F/rd$, we finally get that
	\begin{equation*}
		\left\|T-\sum_{m=1}^{r} F_{\wt{S}_m}\wt{D}_{m}F_{\wt{S}_m}^{*}\right\|_2 
		\leq \wt{O}(1) \|T-T_k\|_2 + \delta \|T\|_{F}.
	\end{equation*}
	Defining $\wt{T}$ as $\wt{T} = \sum_{m=1}^{r} F_{\wt{S}_m}\wt{D}_{m}F_{\wt{S}_m}^{*}$, which has rank at most $r(O(\log d+\log(1/\delta))=\wt{O}(k\log(1/\delta))$, we find that $\wt{T}$ satisfies the claim of the theorem.
\end{proof}
Next, we state the our main result on the existence of a near optimal Toeplitz low-rank approximation in the Frobenius norm.
\begin{reptheorem}{frobeniusexistence}
	Given any PSD Toeplitz matrix $T\in \mathbb{R}^{d\times d}$, $0<\eps,\delta<1$ and an integer $k\leq d$, let $r_1=O(k\log^8d/\eps)$ and $r_2 = O(\log d +\log(1/\delta))$. Then there exists a symmetric Toeplitz matrix $\wt{T} = F_{\wt{S}}\wt{D}F_{\wt{S}}^{*}$ of rank $r=2r_1r_2=\wt{O}((k/\eps)\log(1/\delta))$ such that,
	\begin{enumerate}
		\item $\|T-\wt{T}\|_F \leq (1+\eps)\|T-T_k\|_F + \delta\|T\|_F$, where $T_k = \underset{B:rank(B)\leq k}{\|T-B\|_F}$ is the best rank-$k$ approximation to $T$ in the Frobenius norm.
		\item \label{thm:existence_freq_grid}$F_{\wt{S}},\wt{D}$ are Fourier and diagonal matrices respectively. The set of frequencies $\wt{S}$ can be partitioned into $r_1$ sets $\wt{S}_{1},\ldots, \wt{S}_{r_1}$ where each $\wt{S}_i$ is as follows:
		\begin{equation*}
			\wt{S}_i = \underset{1\leq j\leq r_2}{\bigcup}\{f_i+\gamma j,f_i-\gamma j\}.
		\end{equation*}
		\item  $f_i\in N$  for all $i\in [r_1]$, where $N:=\{1/2d,3/2d,\ldots, 1-1/2d\}$. Furthermore, $\gamma = \delta /(2^{C_2 \log^7d})$, where $C_2>0$ is a fixed constant.
	\end{enumerate}
\end{reptheorem}
\begin{proof}
	Let $T=F_S D F_S$ be $T$'s Vandermonde decomposition and let $\lambda = \lambda_{Ck\log^5 d/\epsilon}(T)\log d/c'd$, where $c'$ is the same constant as in Lemma~\ref{heavybucketlemma} and $C$ is some large constant. Bucket the frequencies in $S$ as per Definition \ref{bucketing}. Let $T_{heavy}$ be the matrix obtained by considering all buckets of $T$ with weight strictly more than $\lambda$. Then $T-T_{heavy}$ is also a Toeplitz matrix defined by the remaining buckets, all of which have weight at most $\lambda$. Thus by Lemma~\ref{lightbucketlemma}, we have the following
	\begin{equation*}
		\|T-T_{heavy}\|_2\leq  O(\lambda d\log d) = O(\log^2d)\lambda_{Ck\log^5 d/\epsilon}(T).
	\end{equation*}
	Since $T-T_{heavy}$ just corresponds to retaining a subset of the frequencies in $S$ and their corresponding weights, it is a PSD change. More formally, $0\preceq T-T_{heavy}\preceq T$. Thus by Weyl's inequality, this implies that $\lambda_i(T-T_{heavy})\leq \lambda_{i}(T)$ for all $i\in [d]$. Thus we have that
	\begin{align*}
		\|T-T_{heavy}\|_F^2 &= \sum_{i=1}^{d} \lambda_{i}^2(T-T_{heavy}) = \sum_{i=1}^{k+1} \lambda_{i}^2(T-T_{heavy}) + \sum_{i=k+2}^{d} \lambda_{i}^2(T-T_{heavy})\\
		&\leq (k+1)\lambda_{1}^2(T-T_{heavy})+ \sum_{i=k+2}^{d} \lambda_{i}^2(T-T_{heavy})\\
		&\leq (k+1)\lambda_{1}^2(T-T_{heavy})+ \sum_{i=k+2}^{d} \lambda_{i}^2(T)    .
	\end{align*}
	Here, we used the fact that $\lambda_i(T-T_{heavy})\leq \lambda_1(T-T_{heavy})$ for all $1\leq i\leq k+1$ and $\lambda_i(T-T_{heavy})\leq \lambda_{i}(T)$ for $i> k+1$. 
	Now, using that $\lambda_1(T-T_{heavy})\leq O(\log^2 d)\lambda_{Ck\log^5 d/\epsilon}(T)$, we get the following:
	\begin{align*}
		\|T-T_{heavy}\|_F^2 &\leq (k+1)\lambda_{1}^2(T-T_{heavy})+ \sum_{i=k+2}^{d} \lambda_{i}^2(T)\\
		&\leq O(k\log^4d)\lambda_{Ck\log^5d/\eps}^2(T) + \sum_{i=k+2}^{d} \lambda_{i}^2(T)\\
		&= \eps O(k\log^4d/\eps)\lambda_{Ck\log^5d/\eps}^2(T) + \sum_{i=k+2}^{d} \lambda_{i}^2(T)\\
		&\leq \eps\sum_{i=k+1}^{O(k\log^4d/\eps)}\lambda_{i}^2(T) + \sum_{i=k+2}^{d} \lambda_{i}^2(T)\\
		&\leq \eps\sum_{i=k+1}^{d}\lambda_{i}^2(T) + \sum_{i=k+1}^{d} \lambda_{i}^2(T)\\
		&= (1+\eps)\|T-T_{k}\|_F^2.
	\end{align*}
	Therefore, we have that $\|T-T_{heavy}\|_F\leq (1+\eps)\|T-T_{k}\|_F$, where $T_{heavy}$ contains all buckets with weight more than $\lambda_{Ck\log^5 d/\eps}\log d/c'd$. There are at most $r_1=O(k\log^8 d/\eps)$ such buckets by Lemma~\ref{heavybucketlemma}.
	Let $F_{S_1},\ldots, F_{S_{r_1}}$ and $D_1,D_2,\ldots, D_{r_1}$ to be the Fourier and diagonal weight matrices corresponding to these $r_1$ heavy buckets defining $T_{heavy}$. We can then write $T_{heavy}$ as follows
	\begin{equation*}
		T_{heavy} = \sum_{m=1}^{r_1} F_{S_m}D_{m}F_{S_m}^{*}.
	\end{equation*}
	
	Since each $ F_{S_m}D_{m}F_{S_m}^{*}$ is $(f,1/d)$-clustered for some $f\in N$, we use Theorem \ref{lem:clusteredcase} to approximate each $ F_{S_m}D_{m}F_{S_m}^{*}$ with a rank $O(\log d+\log(1/\delta))$ symmetric Toeplitz matrix $ F_{\wt{S}_m}\wt{D}_{m}F_{\wt{S}_m}^{*}$ that satisfies 
	\begin{equation*}
		\left\|F_{S_m}D_{m}F_{S_m}^{*} - F_{\wt{S}_m}\wt{D}_{m}F_{\wt{S}_m}^{*}\right\|_{F}\leq \delta(\sum_{i} (D_m)_{i}) + \varepsilon d.
	\end{equation*}
	This implies that
	\begin{align*}
		\left\|T-\sum_{m=1}^{r_1} F_{\wt{S}_m}\wt{D}_{m}F_{\wt{S}_m}^{*}\right\|_F &\leq \left\|T-\sum_{m=1}^{r_1} F_{S_m}D_{m}F_{S_m}^{*}\right\|_F + \left\|\sum_{m=1}^{r_1} (F_{S_m}D_{m}F_{S_m}^{*} - F_{\wt{S}_m}\wt{D}_{m}F_{\wt{S}_m}^{*}) \right\|_F\\
		& = \|T-T_{heavy}\|_F + \left\|\sum_{m=1}^{r_1} (F_{S_m}D_{m}F_{S_m}^{*} - F_{\wt{S}_m}\wt{D}_{m}F_{\wt{S}_m}^{*} )\right\|_F\\
		&\leq \|T-T_{heavy}\|_F + \sum_{m=1}^{r_1} \|F_{S_m}D_{m}F_{S_m}^{*} - F_{\wt{S}_m}\wt{D}_{m}F_{\wt{S}_m}^{*}\|_F\\
		& \leq (1+\epsilon) \|T-T_k\|_F + \delta(\sum_{m,i}(D_m)_{i}) +\varepsilon r d \\
		&\leq  (1+\epsilon) \|T-T_k\|_F + \delta(\sum_{i}(D)_{i}) + \varepsilon r d.
	\end{align*}
	Now using Lemma~\ref{tracetrick}, we get the following:
	\begin{equation*}
		\sum_{i}(D)_{i} = \sum_{i=1}^{d}\lambda_i(T)  /d \leq \left(\sqrt{d \sum_{i=1}^d \lambda_{i}(T)^2}\right)/d = \|T\|_F/\sqrt{d}.
	\end{equation*}
	Setting $\varepsilon = \delta\|T\|_F/rd$, we finally have
	\begin{equation*}
		\left\|T-\sum_{m=1}^{r_1} F_{\wt{S}_m}\wt{D}_{m}F_{\wt{S}_m}^{*}\right\|_F 
		\leq (1+\epsilon) \left\|T-T_k\|_2 + \delta \|T\right\|_{F} .
	\end{equation*}
	Defining $\wt{T}$ as $\wt{T} = \sum_{m=1}^{r_1} F_{\wt{S}_m}\wt{D}_{m}F_{\wt{S}_m}^{*}$, which has rank at most $r_1(C(\log d+\log(1/\delta))=\wt{O}((k/\epsilon)\log(1/\delta))$, we get that $\wt{T}$ satisfies the claim and point 1 of Theorem \ref{frobeniusexistence}. Observe that the Toeplitz matrix corresponding to any bucket $F_{S_m}D_mF_{S_m}^{*}$ is $(f,1/d)$-clustered for some $f\in N$. Thus, it follows by guarantee 1 of Lemma~\ref{lem:clusteredcase} that each $\wt{S}_m$ is of the form described in point 2 of Theorem \ref{frobeniusexistence}. Since $T_{1,1}\leq \|T\|_F$, with this value of $\varepsilon$ we can choose any $\gamma\leq \delta T_{1,1}/(\tr(T)rd^2 2^{C_1 \log^7 d})=\delta/(2^{C_2 \log^7 d})$ for a large enough constant $C_2$. (Here, we used the fact that $\tr(T)=dT_{1,1}$.) This completes the proof of Theorem \ref{frobeniusexistence}.
\end{proof}


\section{Low-rank approximation with sublinear query complexity.}\label{algorithm}
In this section, we present our main algorithm, Algorithm \ref{sublinearqueryalgo} and prove the corresponding Theorem \ref{sublinearqueryalgo}, which shows that the algorithm outputs a near optimal low-rank approximation to $T$, while reading only sublinearly many entries. 
We treat $\wt{T}$, the low-rank Toeplitz approximation to $T$ (guaranteed to exist from Theorem \ref{frobeniusexistence}), as the true matrix, which we noisily access by reading entries of $T$.

\subsection{Reduction to weighted linear regression.}

The near optimal Toeplitz low-rank approximation $\wt T \in \R^{d \times d}$ guaranteed to exist by Theorem \ref{frobeniusexistence} is of the form $\wt T =F_{\wt{S}}\wt{D}F_{\wt{S}}^{*}$, where $ F_{\wt{S}}\in \mathbb{C}^{d\times r}$ and $\wt{D}\in \mathbb{R}^{r\times r}$ are Fourier and diagonal matrices, respectively, and $r=\wt{O}((k/\epsilon)\log(1/\delta))$. Algorithm \ref{sublinearqueryalgo} uses brute force search to find the frequencies in $\wt{S}$. In particular, it uses an approximate regression oracle to test the quality of any guess for the frequencies in $\wt{S}$, without reading many entries of $T$.  

Recall from Theorem \ref{frobeniusexistence} that these frequencies lie in $r_1 = \tilde O(k/\epsilon)$  clusters of $r_2$ equispaced frequencies centered around points in $N =\{1/2d,3/2d,\ldots, 1-1/2d\}$. Thus, our search space will be all subsets of $r_1$ elements of $N$, of which there at at most $N^{r_1}$. Formally, the frequencies of our Toeplitz low-rank approximation will lie in the following set:
\begin{definition}[Frequency Search Space]\label{def:search}
	Consider any positive integers $d, r_1,r_2$ with $r_1,r_2 < d$ and $\gamma \in (0,1)$, 
	Let $N = \{1/2d,3/2d,\ldots,1-1/2d\}$. For any set of frequencies $B$, let $S(B) = \bigcup_{b \in B} \bigcup_{1 \le j \le r_2} \{ b + \gamma j, b - \gamma j\}$. Let $\mathcal{N}_{d,r_1,r_2,\gamma} = \{S(B): B \in N^{r_1}\}$. 
\end{definition}

Since the first column of $\wt{T}$ defines the full matrix, the approximate regression oracle will simply attempt to fit the first column of $\wt{T}$ to be close to that of $T$. However, since different entries in the first column appear with different frequencies in the matrix, we require the following weighting function to translate error bounds in the first column into error bounds for approximating the entire matrix ${T}$ in the Frobenius norm. 
\begin{definition}\label{weightingdef}
	Let $w\in \mathbb{R}^{d}$ be defined as follows:
	\begin{equation*}
		w_{i} =  \begin{cases} 
			\sqrt{d} & i = 1 \\
			\sqrt{2(d-i+1)} & i > 1 
		\end{cases}
	\end{equation*}
	Let $W = \diag(w)\in \mathbb{R}^{d\times d}$.
\end{definition}
We have the following immediate claim, which expresses the Frobenius norm difference between two symmetric Toeplitz matrices as the weighted $\ell_2$ norm difference of their first columns. 
\begin{claim}\label{frobtoeuclid}
	Let $T,  \wt T\in \mathbb{R}^{d\times d}$ be symmetric Toeplitz matrices with first columns $T_{1}, \wt T_1 \in \mathbb{R}^{d}$  respectively. Then letting $W \in \R^{d \times d}$ be as in Definition \ref{weightingdef},  $\|T -  \wt T\|_F = \|WT_1 - W \wt T_1\|_2$. 
\end{claim}

Now, for a Toeplitz matrix $T = F_S D F_S^*$, the first column can be expressed as $T_1 = F_S a$ where $a \in \R^{r}$ contains the diagonal entries of $D$. Further, if $T$ is real, $a$ must place equal weight on the conjugate frequencies in $F_S$. Thus, we can in fact write $T_1 = F_S R_S a$ where $a \in \R^{r/2}$ and $R_S \in \R^{r \times r/2}$ collapses the $2|S|=r$ conjugate pair columns of $F_S$ into $|S|=\frac{r}{2}$ real-valued columns, each corresponding to a degree of freedom of $a$. More formally:
\begin{definition}\label{def:RS} 
	Let $S \subset [0,\frac{1}{2})$ with $|S|=\frac{r}{2}$. 
	Then, define the matrix $R_S \in \mathbb{R}^{r \times \frac{r}{2}}$ by setting the $j^{th}$ column $(R_S)_{:,j}$ equal to $0$ everywhere except at $j$ and $|S|+j$, corresponding to the $j^{th}$ \emph{pair} of conjugate frequencies, where it is equal to $1$. 
\end{definition}



From Claim \ref{frobtoeuclid} and the existence proof of Theorem \ref{frobeniusexistence} we have the following lemma, which shows that fitting $\wt T$ can be reduced to a  real-valued linear regression problem:
\begin{lemma}[Toeplitz Approximation via Weighted Regression]\label{lem:regressionReduction}
	For any PSD Toeplitz matrix $T\in \mathbb{R}^{d\times d}$ with first column $T_1$, $\eps,\delta \in (0,1)$, and integer $k\leq d$, let $\mathcal{N}_{d,r_1,r_2,\gamma}$ be as in Def. \ref{def:search}, where $r_1,r_2,\gamma$ are as in Theorem \ref{frobeniusexistence}. Let $W \in \R^{d\times d}$ be as in Def. \ref{weightingdef}. Then if $\wt S$ and $\wt a \in \R^{|\tilde S|}$ satisfy for some $\alpha \ge 1$:
	\begin{align*}
		\norm{WF_{\wt S} R_{\wt S} \wt a - WT_1}_2 \le \alpha \cdot \min_{S \in \mathcal{N}_{d,r_1,r_2,\gamma}, a \in \R^{|S|}} \norm{WF_SR_S a - WT_1}_2,
	\end{align*}
	letting $\wt T \in \R^{d \times d}$ be the symmetric Toeplitz matrix $\wt T = F_{\wt S} \diag(\wt a) F_{\wt S}^*$,
	$$\norm{T-\wt T}_F \le (1+\epsilon)\alpha \cdot \norm{T-T_k}_F + \alpha \delta \norm{T}_F.$$
\end{lemma}
\begin{proof}
	Letting $\wt T_1$ be the first column of $\wt T$, we have $\wt T_1 = F_{\wt S} R_{\wt S} \wt a$ where $R_{\wt S}$ is defined as in Def. \ref{def:RS}.
	Thus, by Claim \ref{frobtoeuclid},
	\begin{align*}
		\|T -  \wt T\|_F = \|WT_1 - W F_{\wt S}R_{\wt S } \wt a\|_2 &\le  \alpha \cdot \min_{S \in \mathcal{N}_{d,r_1,r_2,\gamma}, a \in \R^{|S|}} \norm{WT_1 - WF_SR_S a}_2\\
		&\le (1+\epsilon) \alpha \cdot \norm{T-T_k}_F + \alpha \delta \norm{T}_F,
	\end{align*}
	where the last inequality follows from Theorem \ref{frobeniusexistence}, which shows that there is some Toeplitz matrix $\wt T$ with frequency set in $\mathcal{N}_{d,r_1,r_2,\gamma}$ satisfying $\norm{T-\wt T}_F \le (1+\epsilon) \norm{T-T_k}_F + \delta \norm{T}_F$.
\end{proof}

\subsection{Leverage score preliminaries.}
Our goal  is now  to find $\wt S \in \mathcal{N}_{d,r_1,r_2,\gamma}$ and $\wt a \in \R^{|\wt S|}$ satisfying the approximate regression guarantee of Lemma \ref{lem:regressionReduction} for $\alpha \le 1+\epsilon$. We do this using leverage score sampling. It is well known that leverage score sampling can be used to approximately solve linear regression problems in a sample-efficient manner \cite{Sarlos:2006vn,Woodruff:2014tg}. In our setting, an additional challenge arises in that to find $\wt S,\wt a$ satisfying the bound of Lemma \ref{lem:regressionReduction} we must solve many regression problems -- corresponding to all possible subsets $\wt S \in \mathcal{N}_{d,r_1,r_2,\gamma}$ -- and output one with near minimal error. This is challenging, since standard results on leverage score sampling for sample-efficient regression 1) typically do not succeed with very high probability, making it difficult to union bound over all possible $\wt S$, and 2) typically do not output an estimate of the actual regression error, making it impossible to chose a near optimal $\wt S,\wt a$ as needed in Lemma \ref{lem:regressionReduction}. We show how to overcome these issues in Sections \ref{sec:const} and \ref{sec:eps}.

Recall the notion of matrix leverage scores as defined in Definition \ref{leveragescoresdefn}. Our algorithm will sample rows of $WF_SR_S$ and $WT_1$ via approximations to the leverage scores of $WF_SR_S$ to solve the regression problem of Lemma \ref{lem:regressionReduction}. For technical reasons, we will sample with a mixture of a leverage score distribution and the uniform distribution, defined below:
\begin{definition}[Leverage score sampling; Def. 2.7 of \cite{musco2021active}]\label{def:levsampling}
	For any number of samples $m$ and a set of leverage score bounds $\{\tilde \tau_j \}_{j \in [n]}$ with $T = \sum_{j=1}^d \wt{\tau}_j$, let $\mathcal{S} \in \mathbb{R}^{m \times d}$ be a sampling matrix with each row set independently to the $i^{th}$ standard basis vector multiplied by $(m \cdot p_i)^{-1/2}$, with probability 
	\begin{align*}
		p_i &= \frac{1}{2} \Bigg( \frac{\wt{\tau}_i}{T} + \frac{1}{d} \Bigg).
	\end{align*}
\end{definition}
We also have the following unbiasedness property of the sampling matrix $\mcS$.
\begin{claim}\label{claim:unbiased}
	Let $\mcS$ be defined as in \ref{def:levsampling}. Then
	\begin{align*} 
		\E[\norm{\mcS x}_2^2] &= \norm{x}_2^2.
	\end{align*}
\end{claim}
Critically, we would like to take \emph{a single set} of samples and use them to approximately minimize $\norm{WT_1 - WF_SR_S a}_2$ for all $S \in \mathcal{N}_{d,r_1,r_2,\gamma}$ in order to solve the optimization problem of Lemma \ref{lem:regressionReduction}. We are able to do this due to the existence of \emph{universal} leverage score bounds on Fourier matrices \cite{Avron:2019tt,eldar2020toeplitz}, which are independent of the frequency subset $S$. We adapt these bounds to our  weighted setting. 
\begin{lemma}[Fourier Matrix Leverage Score Bound]\label{leveragescorebound}
	Let $W$ be the weight matrix defined in \ref{weightingdef}. Then there exist non-negative numbers $\{\wt{\tau}_{j}\}_{j=1}^d$ such that the following hold for all frequency sets $S\subset [0,\frac{1}{2}), |S| = \frac{r}{2}\leq d$: 
	\begin{enumerate}
		\item $\tau_j(WF_SM_S) \leq \tau_j(WF_S) \leq \wt{\tau}_{j}$ for all $j\in [d]$.
		\item $\sum_{j=1}^{d} \wt{\tau}_{j} \leq O(r\log r \cdot \log d)$.
	\end{enumerate}
	Here, $F_S\in \mathbb{C}^{d\times r}$ is the symmetric Fourier matrix with frequency set $S$ as defined in \ref{def:symmetric_fourier_matrix}, and $M_S$ is any real-valued matrix with $r$ rows, which may depend on $S$. (For example, this includes the case $M_S=R_S$ as defined in Definition \ref{def:RS}.) Finally, $\tau_j(\cdot)$ is as defined in Def. \ref{leveragescoresdefn}.
\end{lemma}
\ref{leveragescorebound} is similar to Corollary C.2 of \cite{eldar2020toeplitz}, but applies to \emph{weighted} Fourier matrices. In \ref{subsec:weighted_fourier_lev_bounds} we restate Corollary C.2 of \cite{eldar2020toeplitz} as \ref{corr:C2_eldar}, and use it to prove  \ref{leveragescorebound}.

In our analysis, we will use the well known fact  that sampling $\tilde O(r)$ rows of $A \in \mathbb C^{d \times r}$ according to their leverage scores yields a \emph{subspace} embedding of $A$, which preserves the norms of all vectors in $A$'s column span to small relative error. In particular:
\begin{lemma}[Subspace Embedding \cite{Woodruff:2014tg}]\label{lem:subspace}
	Given $A \in \R^{d \times r}$, failure probability $\eta \in (0,1)$, and a set of leverage score upper bounds $\{\tilde \tau_j \}_{j \in [n]}$ satisfying $\tilde \tau_j \ge \tau_j(A)$ for all $j \in [d]$, let $\mathcal{S} \in \mathbb{R}^{m \times d}$ be a random sampling matrix drawn according to \ref{def:levsampling}, using the distribution $\wt{\tau}_j$ and $m=O\big(\frac{T \log(1/\eta)}{\beta^2}\big)$ samples. Then, with probability $\ge 1-\eta$, for all $x \in \R^r$, 
	$$(1-\beta) \norm{ Ax}_2 \le \norm{\mathcal{S}  Ax}_2 \le (1+\beta) \norm{ Ax}_2.$$
\end{lemma}

\subsection{Constant factor approximation.}\label{sec:const}
We now use the leverage score bounds of Lemma \ref{leveragescorebound} and the subspace embedding guarantee of Lemma \ref{lem:subspace} to show how to solve the optimization problem of Lemma \ref{lem:regressionReduction} for constant $\alpha$. We later show how to refine this to an $\alpha = 1+\epsilon$ approximation, achieving our final bound.

\begin{theorem}[Constant Factor Toeplitz Fitting]\label{thm:const}
	Consider the setting of Lemma \ref{lem:regressionReduction}. Let
	$\mathcal{S} \in \mathbb{R}^{m \times d}$ be a random sampling matrix drawn according to \ref{def:levsampling}, using the distribution $\wt{\tau}_j$ defined in Def. \ref{leveragescorebound} with $m = \tilde O \left( \frac{k^2\log(1/\delta)}{\epsilon^2} \right )$. Let
	\begin{align*}
		\tilde S, \tilde a = \argmin_{S \in \mathcal{N}_{d,r_1,r_2,\gamma}, a \in \R^{|S|}} \norm{\mathcal{S}WF_SR_S a - \mathcal{S}WT_1}_2.
	\end{align*}
	Then with probability at least $98/100$, $\tilde S$ and $\tilde a$ satisfy Lemma \ref{lem:regressionReduction} with $\alpha = 41$. In other words,
	\begin{align*}
		\norm{WF_{\wt S} R_{\wt S} \wt a - WT_1}_2 \le 41 \cdot \min_{S \in \mathcal{N}_{d,r_1,r_2,\gamma}, a \in \R^{|S|}} \norm{WF_SR_S a - WT_1}_2.
	\end{align*}
\end{theorem}
\begin{proof}
	Define the optimal frequency set and coefficients by $$S', a' = \text{argmin}_{S \in \mathcal{N}_{d,r_1,r_2,\gamma}, a \in \R^{|S|}} \norm{WF_SR_S a - WT_1}_2.$$ For any $S\in \mathcal{N}_{d,r_1,r_2,\gamma}, a \in \R^{|S|}$ we can write $$\|WF_SR_Sa - WF_{S'}R_{S'}a'\|_2$$ as $$\|W[F_SR_S,F_{S'}R_{S'}][a;-a']\|_2 = \|W[F_S,F_{S'}][R_S;R_{S'}][a;-a']\|_2.$$ Here, $[F_S,F_{S'}]$ is a Fourier matrix, and $[R_S;R_{S'}]$ fits the criteria of the matrix $M_{S,S'}$ in \ref{leveragescorebound}. Therefore, by \ref{leveragescorebound} the leverage scores  of $W[F_S,F_{S'}][R_S;R_{S'}]$ are upper bounded by $\tilde{\tau}_j$ and sum to $O(r\log(d)\log(r))$. Applying \ref{lem:subspace} with $A=W[F_S,F_{S'}][R_S;R_{S'}]$, $\beta=\frac{1}{2}$, $\eta = 1/(100|N|^{r_1})$, and the upper bounds $\tilde{\tau}_j$, we obtain a sampling matrix $\mcS\in \mathbb{R}^{m\times d}$ that takes $m=\wt{O}(\frac{k^2\log(1/\delta)}{\epsilon^2})$ samples. By combining the probabilistic guarantee of \ref{lem:subspace} with a union bound over all $S \in \mathcal{N}_{d,r_1,r_2,\gamma}$, we get that the following holds with probability at least $1-|N|^{r_1}\eta = \frac{99}{100}$ over the choice of $\mcS$:
	\begin{align*}
		\|\mcS W [F_SR_S,F_{S'}R_{S'}] x \|_2 &\geq \frac{1}{2} \| W [F_SR_S,F_{S'}R_{S'}] x \|_2 \: \: \:  \forall S \in \mathcal{N}_{d,r_1,r_2,\gamma}, \: x \in \R^{r}.
	\end{align*} 
	This implies for the particular case of $S=\tilde{S}$ and $a=\tilde{a}$ that the following holds with probability at least $\frac{99}{100}$:
	\begin{align*}
		&\leq 2\|\mcS W [F_{\tilde{S}}R_{\tilde{S}},F_{S'}R_{S'}] [\ta;-a'] \|_2 \\
		&= 2\|\mcS WF_{\tilde{S}}R_{\tilde{S}}\tilde{a} - \mcS WF_{S'}R_{S'}a' \|_2  \\
		&\leq 2 (\|\mcS WF_{\tilde{S}}R_{\tilde{S}}\tilde{a} - \mcS WT_1 \|_2 + \|\mcS WT_1 - \mcS WF_{S'}R_{S'}a' \|_2  ) \\
		&\leq 2 (\| \mcS WF_{S'}R_{S'}a' - \mcS WT_1  \|_2 + \| \mcS WF_{S'}R_{S'}a' - \mcS WT_1  \|_2  ) \\
		&= 4 \| \mcS WF_{S'}R_{S'}a' - \mcS WT_1  \|_2.
	\end{align*}
	The last inequality above followed by the definition of $\tilde{S}$ and $\tilde{a}$ as minimizing $\norm{WF_SR_S a - WT_1}_2$ over all $S \in \mathcal{N}_{d,r_1,r_2,\gamma}, \: a \in \R^{r}$.  
	
	Moreover, by \ref{claim:unbiased}, $\E[\norm{\mcS WF_{S'}R_{S'}a' - \mcS WT_1}_2^2] = \norm{WF_{S'}R_{S'}a' - WT_1}_2^2$. Then by applying Markov's inequality, $\norm{\mcS WF_{S'}R_{S'}a' - \mcS WT_1}_2^2 \leq 100\norm{WF_{S'}R_{S'}a' - WT_1}_2^2$ with probability at least $\frac{99}{100}$. 
	Finally, we return to the quantity of interest, $\norm{WF_{\tilde{S}}R_{\tilde{S}}\tilde{a} - WT_1}_2$. By applying a union bound once more, the following then holds with probability at least $\frac{98}{100}$ over the choice of $\mcS$:
	\begin{align*}
		\norm{WF_{\tilde{S}}R_{\tilde{S}}\tilde{a} - WT_1}_2 &\leq \norm{WF_{S'}R_{S'}a' - WT_1}_2 + \norm{WF_{\tilde{S}}R_{\tilde{S}}\tilde{a} - WF_{S'}R_{S'}a'}_2  \\
		& \leq \norm{WF_{S'}R_{S'}a' - WT_1}_2 + 4 \| \mcS WF_{S'}R_{S'}a' - \mcS WT_1  \|_2  \\
		&\leq \norm{WF_{S'}R_{S'}a' - WT_1}_2 + 40 \norm{WF_{S'}R_{S'}a' - WT_1}_2  \\
		& \leq 41 \norm{WF_{S'}R_{S'}a' - WT_1}_2\\
		& = 41 \min_{S \in \mathcal{N}_{d,r_1,r_2,\gamma}, a \in \R^{|S|}} \norm{WF_SR_S a - WT_1}_2.
	\end{align*}
	This concludes the proof of the lemma.
\end{proof}

\subsection{$(1+\epsilon)$-approximation.}\label{sec:eps}

Theorem \ref{thm:const} combined with Lemma \ref{lem:regressionReduction} yields a $\tilde O \left (\frac{k^2\log(1/\delta)}{\epsilon^2}\right )$ query algorithm for outputting $\tilde T$ with rank $\tilde O(k \log(1/\delta)/\epsilon)$ and $\norm{T-\tilde T}_F = O(1)  \norm{T-T_k}_F + \delta \norm{T}_F$. To prove Theorem \ref{sublinearqueryalgo} we need to improve this constant factor approximation to $(1+\epsilon)$.
We do this using recently developed guarantees for high probability relative error active regression via leverage score sampling \cite{musco2021active}. Importantly, we first compute a constant error solution via Theorem \ref{thm:const}. We then show that we can fit the \emph{residual} of this approximation to high accuracy via leverage score sampling. 
Formally,
\begin{theorem}\label{thm:relative}
	Consider the setting of Lemma \ref{lem:regressionReduction} and let $\tilde S \in \mathcal{N}_{d,r_1,r_2,\gamma}$ and $\tilde a \in \R^{|\tilde S|}$ satisfy the lemma with $\alpha = O(1)$. Let $T_R = T_1 - F_{\tilde S} R_{\tilde S} \tilde a$ be their residual  in fitting the first column of $T$.
	
	Let
	$\mathcal{S} \in \mathbb{R}^{m \times d}$ be a random sampling matrix drawn according to \ref{def:levsampling}, using the distribution $\wt{\tau}_j$ defined in Def. \ref{leveragescorebound} with $m =  \tilde{O}\Big( \frac{k^2\log(\frac{1}{\delta})}{\epsilon^{6}} \Big)$. Let
	\begin{align*}
		S', a' = \argmin_{S \in \mathcal{N}_{d,2r_1,r_2,\gamma}, a \in \R^{|S|}} \norm{\mathcal{S}WF_SR_S a - \mathcal{S}WT_R}_2.
	\end{align*}
	Then with probability at least $99/100$, letting $\bar S = \tilde S \cup S'$ and $\bar a = [\tilde a, a']$, we have that $\bar S,\bar a$ satisfy Lemma \ref{lem:regressionReduction} with $\alpha = (1+\epsilon)$. I.e.,
	\begin{align*}
		\norm{WF_{\bar S} R_{\bar S} \bar a - WT_1}_2 \le (1+\epsilon) \cdot \min_{S \in \mathcal{N}_{d,r_1,r_2,\gamma}, a \in \R^{|S|}} \norm{WF_SR_S a - WT_1}_2,
	\end{align*}
	and thus, letting $\bar T= F_{\bar S} \diag(\bar a) F_{\bar S}^*$,
	\begin{align*}
		\norm{T-\bar T}_F \le (1+3\epsilon) \norm{T-T_k}_F + 2 \delta \norm{T}_F.
	\end{align*}
\end{theorem}
Note that $\bar  S$ output by Theorem \ref{thm:relative} does not lie in $\mathcal{N}_{d,r_1,r_2,\gamma}$. Since $S' \in \mathcal{N}_{d,2r_1,r_2,\gamma}$ and $\tilde S \in \mathcal{N}_{d,r_1,r_2,\gamma}$, we have $\bar S \in \mathcal{N}_{d,3r_1,r_2,\gamma}$. Lemma \ref{lem:regressionReduction} allows this -- it simply means that the rank of the corresponding Toeplitz matrix $\bar T$ may be three times as large as if $\bar S$ were in $\mathcal{N}_{d,r_1,r_2,\gamma}$.  
\begin{proof}
	Let $OPT = \min_{S \in \mathcal{N}_{d,r_1,r_2,\gamma}, a \in \R^{|S|}} \norm{WF_SR_S a - WT_1}_2$. Observe that 
	\begin{align*}
		\norm{WF_{\bar S} R_{\bar S} \bar a - WT_1}_2 = \norm{WF_{S'} R_{S'} a' + WF_{\tilde  S} R_{\tilde S} \tilde  a- WT_1}_2 = \norm{WF_{S'} R_{S'} a' - WT_R}_2.
	\end{align*}
	Thus, to prove the theorem, it suffices to show that
	\begin{align}\label{eq:sufficesBound1}
		\norm{WF_{S'} R_{S'} a' - WT_R}_2 \le (1+\epsilon) \cdot OPT.
	\end{align}
	Further note that by the assumption that $\tilde S,\tilde a$ satisfy Lemma \ref{lem:regressionReduction} with $\alpha = O(1)$ we have:
	\begin{align*}
		\norm{WT_R}_2 \le \alpha \cdot OPT = O(OPT).
	\end{align*}
	By Markov's inequality, since by construction $\E[\norm{\mathcal S WT_R}_2^2] = \norm{W T_R}_2^2$, with probability at least $999/1000$, this also gives that $\norm{\mathcal{S}WT_R}_2 \le \sqrt{1000} \cdot \alpha \cdot OPT =  O(OPT)$. 
	
	For a given $S \in \mathcal{N}_{d,2r_1,r_2,\gamma}$ define:
	\begin{align*}
		&a_S = \argmin_{a \in \R^{|S|}} \norm{W F_S R_S a - W T_R}_2 \\ &\text{ and }\\
		&a'_S = \argmin_{a \in \R^{|S|}} \norm{\mathcal{S} W F_S R_S a - \mathcal{S}W T_R}_2.
	\end{align*}
	Using that $\norm{W T_R}_2 = O(OPT)$ and applying triangle inequality, we have
	\begin{align}\label{eq:solutionNormBound1}
		\norm{WF_SR_Sa_S - WT_R}_2 &\le \norm{WT_R}_2\nonumber\\
		\norm{WF_SR_Sa_S}_2 - \norm{WT_R}_2 &\le \norm{WT_R}_2\nonumber\\
		\norm{WF_SR_Sa_S}_2 &= O(OPT).
	\end{align}
	Similarly, using that $\norm{SWT_R}_2 = O(OPT)$ with good probability, for all $S$ we have $\norm{\mathcal S WF_SR_Sa'_S}_2 = O(OPT)$. Further, as in Theorem \ref{thm:const}, with probability at least $1-|N|^{2r_1}\eta = 99/100$ for $\eta=1/(100|N|^{2r_1})$, the subspace embedding guarantee of Lemma \ref{lem:subspace} holds for all $WF_S R_S$ simultaneously for $\beta = O(1)$, giving  that for all $S \in \mathcal{N}_{d,2r_1,r_2,\gamma}$
	\begin{align}\label{eq:solutionNormBound2}
		\norm{WF_SR_Sa'_S}_2 = O(OPT).
	\end{align}
	
	Given \eqref{eq:solutionNormBound1} and \eqref{eq:solutionNormBound2}, to prove the theorem it suffices to show the following claim:
	\begin{claim}\label{clm:mainClaim}
		With probability at least $99/100$, for any $S \in \mathcal{N}_{d,2r_1,r_2,\gamma}$ and any $a \in \R^{|S|}$ with $\norm{WF_SR_Sa}_2 = O(OPT)$,
		\begin{align*}
			\left | \norm{WF_SR_Sa - W T_R}_2^2 - \norm{\mathcal S WF_SR_Sa - \mathcal S W T_R}_2^2 - C \right | \le \epsilon \cdot OPT^2,
		\end{align*}
		where $C$ is a fixed constant that may depend on $\mathcal S$ and $WT_R$, but does not depend on $S$ or $a$.
	\end{claim}
	I.e., for any frequency set $S$ and coefficient vector $a \in \R^{|S|}$ where $\norm{WF_SR_Sa}_2 = O(OPT),$ the sampled regression cost, after shifting by a fixed constant, approximates the true regression cost up to additive error $\epsilon \cdot OPT^2$. This ensures that:
	\begin{align*}
		\norm{WF_{S'} R_{S'} a' - WT_R}_2^2 &= \norm{WF_{S'} R_{S'} a'_{S'} - WT_R}_2^2\\
		&\le \norm{\mathcal S WF_{S'}R_{S'}a'_{S'} - \mathcal S W T_R}_2^2 + C + \epsilon \cdot OPT^2\\
		&\le \norm{\mathcal S WF_{S^*}R_{S^*}a^* - \mathcal S W T_R}_2^2 + C + \epsilon \cdot OPT^2,
	\end{align*}
	where $S^*,a^* = \argmin_{S \in \mathcal{N}_{d,2r_1,r_2,\gamma}, a \in \R^{|S|}} \norm{WF_{S}R_{S}a -  W T_R}_2$. The first inequality follows from Claim \ref{clm:mainClaim}, which can be applied since $\norm{WF_{S'}R_{S'}a'_{S'}}_2 = O(OPT)$ by \eqref{eq:solutionNormBound2}. The second inequality follows since $S', a'_{S'} = \argmin_{S \in \mathcal{N}_{d,2r_1,r_2,\gamma}, a \in \R^{|S|}} \norm{\mathcal{S}WF_SR_S a - \mathcal{S}WT_R}_2$. Applying Claim \ref{clm:mainClaim} again to $S^*,a^*$, which is valid by \eqref{eq:solutionNormBound1}, we continue to bound:
	\begin{align}\label{eq:camInter}
		\norm{WF_{S'} R_{S'} a' - WT_R}_2^2 &\le \norm{\mathcal S WF_{S^*}R_{S^*}a^* - \mathcal S W T_R}_2^2 + C + \epsilon \cdot OPT^2\nonumber\\
		&\le \norm{WF_{S^*}R_{S^*}a^* - W T_R}_2^2 + 2 \epsilon \cdot OPT^2.
	\end{align}
	Finally, note that since we allow $S^* \in \mathcal{N}_{d,2r_1,r_2,\gamma},$ we have:
	\begin{align*}
		\norm{WF_{S^*}R_{S^*}a^* - W T_R}_2^2 &\le \min_{S \in  \mathcal{N}_{d,r_1,r_2,\gamma},a \in \R^{|S|}} \norm{WF_SR_S a + WF_{\tilde S} R_{\tilde S} \tilde a- W T_R}_2^2\\
		&= \min_{S \in  \mathcal{N}_{d,r_1,r_2,\gamma},a \in \R^{|S|}} \norm{WF_SR_S a - W T_1}_2^2 = OPT^2.
	\end{align*}
	Combined with \eqref{eq:camInter}, this gives that $\norm{WF_{S'} R_{S'} a' - WT_R}_2^2 \le (1 + 2\epsilon) \cdot OPT^2$, which, after taking a square root and adjusting $\epsilon$ by a constant yields \eqref{eq:sufficesBound1} and in turn the theorem.
	%
	\end{proof}
	We finally present the proof of Claim \ref{clm:mainClaim} below.
	\begin{proof}
		Claim \ref{clm:mainClaim} can be proven following the same approach as Theorem 3.4 of \cite{musco2021active}. For simplicity of notation, let $z := WT_R$. We define a set of `bad indices' where the relative size of $z_j$ is significantly larger than the leverage score $\tilde \tau_j$. The regression error on these bad indices will not be well approximated via leverage score sampling. However, this is ok, since no $W F_S R_S a$ can do a good job fitting these indices, given the leverage score bounds on $WF_S R_S$. Formally, let
		\begin{align*}
			\mathcal{B} = \left \{j \in [d] : \frac{z_j^2}{OPT^2} \ge \frac{\tilde \tau_j}{\epsilon^2} \right \}.
		\end{align*}
		Further, we let $\bar z \in \R^d$ be equal to $z$, except with  $\bar z_j = 0$ for all $j \in \mathcal{B}$. Importantly $\mathcal{B}$ and $\bar z$ are defined independently of any specific $S$. For any $S, a \in \R^{|S|}$ with $\norm{WF_SR_Sa}_2 = O(OPT)$ and any $j \in \mathcal{B}$, we have by the definition of the leverage score (Def. \ref{leveragescoresdefn}),
		\begin{align*}
			\left | (WF_SR_S a)_j \right |^2 &\le \norm{WF_SR_SA}_2^2 \cdot {\tilde \tau_j}\\
			&\le O(OPT^2) \cdot \frac{\epsilon^2 \cdot |z_j|^2}{OPT^2}\\
			&= O(\epsilon^2 \cdot |z_j|^2).
		\end{align*}
		This gives that 
		\begin{align*}
			\left | (WF_SR_S a)_j - z_j \right |^2 - \left | (WF_SR_S a)_j - \bar z_j \right |^2 &= \left | (WF_SR_S a)_j - z_j \right |^2 - \left | (WF_SR_S a)_j \right |^2 = (1 \pm O(\epsilon)) \cdot z_j^2.
		\end{align*}
		Thus, for $C_1 = \sum_{j \in \mathcal{B}} |z_j|^2 = \norm{z-\bar z}_2^2$, 
		\begin{align}\label{eq:zBar1Bound}
			\left | \norm{WF_SR_S a - z}_2^2 - \norm{WF_SR_S a - \bar z}_2^2 - C_1 \right | = O(\epsilon) \cdot OPT^2,
		\end{align}
		where we use that $\norm{z}_2^2 = O(OPT^2)$.
		Using the same proof, and the fact that with probability at least $999/1000$ by Markov's inequality, $\norm{\mathcal{S}z}_2^2 = O(OPT^2)$, we have for $C_2 = \norm{\mathcal S (z-\bar z)}_2^2$,
		\begin{align}\label{eq:zBar2Bound}
			\left | \norm{\mathcal S WF_SR_S a - \mathcal S z}_2^2 - \norm{\mathcal S WF_SR_S a - \mathcal S \bar z}_2^2 - C_2 \right | = O(\epsilon) \cdot OPT^2.
		\end{align}
		Observe that $C_1$ and $C_2$ only depend on the sampling matrix $\mathcal{S}$ and truncated vector $\bar z$, whose definition is independent of any specific frequency set $S$ or coefficient vector $a$. Thus, 
		with \eqref{eq:zBar1Bound} and \eqref{eq:zBar2Bound} in place, to prove Claim \ref{clm:mainClaim} it suffices to show that, for all $S \in  \mathcal{N}_{d,r_1,r_2,\gamma}$ and  $a \in \R^{|S|}$ with $\norm{WF_SR_Sa}_2 = O(OPT)$,
		\begin{align}\label{eq:bernsteinError}
			\left | \norm{WF_SR_Sa - \bar z}_2^2 - \norm{\mathcal S WF_SR_Sa - \mathcal S \bar z}_2^2  \right | = O(\epsilon)  \cdot OPT^2.
		\end{align}
		Observe that by definition of $\mathcal{B}$, the entries of $\bar z$ are bounded by $\bar z_j^2 \le OPT^2 \cdot \frac{\tilde \tau_j}{\epsilon^2}$. Similarly, since by assumption $\norm{WF_SR_Sa}_2^2 = O(OPT^2)$ and by the definition  of the leverage scores (Def. \ref{leveragescoresdefn}) $(WF_SR_Sa)_j^2 \le OPT^2 \cdot \tilde \tau_j$. Thus, sampling entries with probabilities proportional to their leverage scores as in Def. \ref{def:levsampling} ensures that for any fixed $S$ and $a \in \R^{|s|}$, by a standard Bernstein bound, \eqref{eq:bernsteinError} holds with high probability. This bound can then be extended to hold to all $S$ and $a$ via an $\epsilon$-net analysis as follows.
	Fix an $S$. For simplicity we assume by scaling that $\|z\|_2=1$ and $OPT =\Theta(1)$. By Claim 3.8 in the version 1 of \cite{musco2021active}, to prove \eqref{eq:bernsteinError} it suffices to show that the following holds with high probability, for all $y\in \mathcal{N}_{\epsilon}$,
	\begin{equation*}
		|\|\mathcal S y - \mathcal S \bar z\|_2^2 - \|y - \bar z\|_2^2 | \leq \epsilon
	\end{equation*}
where $\mathcal{N}_{\epsilon} $ is an $\epsilon$-net of the set $\{WF_SR_S a: \|WF_SR_S a\|_2\leq 1\}$. By a standard volume argument, it is known that one can construct such a net with $\log |\mathcal{N}_{\epsilon}| = \wt{O}(r_1r_2 \log (1/\epsilon))$. We will show equation \eqref{eq:bernsteinError} for a fixed $S$ and all $a\in \mathbb{R}^{|S|}$ using Bernstein's inequality and a union bound. We have that $\mathbb{E}[\|\mathcal S y - \mathcal S \bar z\|_2^2] = \|y-\bar z\|_2^2 = O(1)$. Additionally by definition $|\bar z_i|^2 \leq  \frac{\wt \tau_i}{\epsilon^2}$ for all $i$. Similarly by the definition of leverage scores $|y_i|^2 \leq \wt \tau_j$. This implies the following,
\begin{equation*}
	|y_i-\bar z_i|^2 = O\left(\frac{\wt \tau_i}{\eps^2}\right)
\end{equation*} 
for all $i$.
By the construction of $\mathcal S$, we have the following,
\begin{equation*}
	|[\mathcal Sy - \mathcal S \bar z]_i|^2 \leq \frac{\sum_i \wt \tau_i}{m \wt \tau_i}\cdot \frac{\wt \tau_i}{\epsilon^2} \leq \wt{O} \left(\frac{r_1r_2}{m\epsilon^2}\right)
\end{equation*}
for all $i$. Thus, by applying a Bernstein bound, we get the following,
\begin{align*}
	Pr[|\|\mathcal S y - \mathcal S \bar z\|_2^2 - \|y - \bar z\|_2^2| > \epsilon] &\leq 2 \exp \left( -\wt \Omega\left(\frac{\epsilon^4 m}{r_1r_2}\right)\right)\\
	& \leq \eta/|\mathcal{N}_{\epsilon}|
\end{align*}
for $m = \wt O (r_1r_2 \log (|\mathcal{N}_{\epsilon}|/\eta)/\epsilon^4)$. Taking a union bound over all $y$ and by Claim 3.8 of \cite{musco2021active}, we get that equation \eqref{eq:bernsteinError} holds for a fixed $S$ and all $a\in \mathbb{R}^{|S|}$ with probability at least $1-\eta$.
Setting $\eta = \frac{1}{100 \cdot |\mathcal{N}_{d,2r_1,r_2,\gamma}|} = \frac{1}{100 \cdot |N|^{2r_1}}$ in the previous corollary, we have that for 
\begin{align*}
	m &= \tilde O \left (\frac{r_1 r_2}{\epsilon^4}\log(|\mathcal{N}_{\epsilon}| \cdot 100|N|^{2r_1})\right )  \\
	& = \tilde O \left (\frac{r_1 r_2}{\epsilon^4}\left (\log(|\mathcal{N}_{\epsilon}|)+ \log( |N|^{2r_1})\right )\right ) \\
	& = \tilde O \left (\frac{r_1 r_2}{\epsilon^4}\left (r_1r_2 + r_1\right )\right )\\
	&= \tilde O \left (\frac{k^2 \log(1/\delta)}{\epsilon^{6}} \right ),
\end{align*}
equation \eqref{eq:bernsteinError} holds for all $S \in \mathcal{N}_{d,2r_1,r_2,\gamma}$ simultaneously with probability at least $99/100$. This completes the proof of the claim.
	\end{proof}

	\subsection{Full algorithm.}
	
	Equipped with these tools, we now describe the recovery algorithm and its guarantees. 
	
	\begin{algorithm}[H]
		\caption{\textsc{ToeplitzRecovery}}
		\label{toeplitzrecovery}
		\begin{algorithmic}[1]
			
			\STATE $\textbf{Input:}$ Query access to $T\in \mathbb{R}^{d\times d}$, $k, \epsilon, \delta$.
			\STATE $\textbf{Init:}$ Set $r_1=O(k\log^8d/\eps)$, $r_2 = O(\log d +\log(1/\delta))$, $\eta=\frac{1}{100|N|^{r_1}}$, $\wt{\tau}_j$ as defined in \ref{leveragescorebound}, $\mathcal{N}_{d,r_1,r_2,\gamma}$ and $\mathcal{N}_{d,2r_1,r_2,\gamma}$ as defined in \ref{def:search}, and $\gamma= \delta / (2^{C_2 \log^7(d)})$ as in \ref{frobeniusexistence}.
			\STATE Draw $\mathcal{S}_1$ according to \ref{def:levsampling}, using $\wt{\tau}_j$ with $m=\tilde{O}\Big( \frac{k^2\log(\frac{1}{\delta})}{\epsilon^2} \Big)$ samples. 
			\STATE Set $\tilde S, \tilde a = \argmin_{S \in \mathcal{N}_{d,r_1,r_2,\gamma}, a \in \R^{|S|}} \norm{\mathcal{S}_1WF_SR_S a - \mathcal{S}_1WT_1}_2$.
			
			\STATE Set $T_R = T_1 - F_{\tilde{S}}R_{\tilde{S}}\ta$ .
			
			\STATE Draw $\mathcal{S}_2$ according to \ref{def:levsampling}, using $\wt{\tau}_j$ with $m=\tilde{O}\Big( \frac{k^2\log(\frac{1}{\delta})}{\epsilon^{6}} \Big)$ samples.
			\STATE Set $S', a' = \argmin_{S \in \mathcal{N}_{d,2r_1,r_2,\gamma}, a \in \R^{|S|}} \norm{\mathcal{S}_2WF_SR_S a - \mathcal{S}_2WT_R}_2$.
			
			\STATE Set $\bar{S} := \tilde{S} \cup S'$, $d := \bar{a} = [\tilde{a}; a']$, and $F = F_{\bar{S}}$.
			\STATE $\textbf{Return:}$ $F,d$.
		\end{algorithmic}
	\end{algorithm}
	
	\begin{reptheorem}{sublinearqueryalgo}
		Assume we are given query access to a PSD Toeplitz matrix $T\in \mathbb{R}^{d\times d}$ and parameters $k,\epsilon, \delta$. Let $r = \wt{O}(\frac{k}{\epsilon}\log(1/\delta))$. Then Algorithm \ref{toeplitzrecovery} returns $F\in \mathbb{C}^{d\times \ell},d\in \mathbb{R}^{\ell}$ such that $F \diag(d)F^{*}$ is a symmetric Toeplitz matrix of rank at most $\ell = O(r)$ and the following holds.
		\begin{enumerate}
			\item Algorithm \ref{toeplitzrecovery} makes $\tilde{O}\Big( \frac{k^2\log(\frac{1}{\delta})}{\epsilon^{6}} \Big)$ queries to $T$. 
			\item $\|T-F\diag(d)F^{*}\|_F\leq (1+3\epsilon)\|T-T_k\|_F + 2\delta \|T\|_F$ with probability at least $97/100$. 
		\end{enumerate}
	\end{reptheorem}
	\begin{proof}
		By \ref{thm:const}, with probability at least $98/100$, $\tilde{S}$ and $\ta$ satisfy \ref{lem:regressionReduction} with $\alpha=41$. The conditions for \ref{thm:relative} are then satisfied, so with conditional probability at least $99/100$, the Toeplitz matrix $\bar{T} := F\diag(d)F^*$ satisfies 
		$$ \norm{T - \bar{T}}_F \leq (1+3\epsilon) \norm{T-T_k}_F + 2 \delta \norm{T}_F.$$
		By a union bound, the matrix $\bar{T} := F\diag(d)F^*$ based on the output of Algorithm \ref{toeplitzrecovery} then satisfies $\norm{T - \bar{T}}_F \leq (1+3\epsilon) \norm{T-T_k}_F + 2 \delta \norm{T}_F$ with probability at least $97/100$. 
		The overall sample complexity is  $\tilde{O}\Big( \frac{k^2\log(\frac{1}{\delta})}{\epsilon^{6}} \Big)$, as this is the sample complexity of $\mcS_2$ and dominates that of $\mcS_1$.  
	\end{proof}

	\subsection{Leverage score bounds for weighted Fourier matrices.}\label{subsec:weighted_fourier_lev_bounds}
	In this subsection, we focus on proving the leverage score upper bounds for weighted Fourier matrices of Lemma \ref{leveragescorebound}. Recall the statement of Lemma \ref{leveragescorebound} was as follows.
	\begin{replemma}{leveragescorebound}
		Let $F_S\in \mathbb{C}^{d\times r}$ be any Fourier matrix with symmetric (in the sense of \ref{thm:vandermonde}) frequency set $S\subset [0,1], |S| = r\leq d$,$r\leq d$, and let $W$ be as defined in \ref{weightingdef} and and $M_S$ be any real valued matrix with $r$ rows. Then there exist non-negative numbers $\wt{\tau}_{j}$ for all $j\in [d]$ such that the following hold:
		\begin{enumerate}
			\item $\tau_j(WF_SM_S) \leq \tau_j(WF_S) \leq \wt{\tau}_{j}$ for all $j\in [d]$.
			\item $\sum_{j=1}^{d} \wt{\tau}_{j} \leq O(r\log (d)\log (r))$.
		\end{enumerate}
	\end{replemma}
	
	To prove this lemma, we will use the following three helper lemmas. 
	\begin{lemma}\label{levscore_reweight}
		Let $A \in \mathbb{C}^{d \times r}$ be a matrix, and let $D$ be a diagonal matrix with positive entries satisfying 
		$$ \alpha \leq D_{ii}^2 \leq \beta .$$
		Then, 
		$$ \tau_i(DA) \leq \frac{\beta}{\alpha}\tau_i(A). $$
	\end{lemma}
	\begin{proof}
		By the definition of leverage scores, we have the following:
		$$\tau_i(DA) = \max_{y\in \mathbb{C}^{r}} \frac{|DAy|_i^2}{\sum_{j=1}^d |DAy|_j^2}. $$
		For any $y$, the numerator is given by $|DAy|_i^2 = D_{ii}^2|Ay|_i^2$. The denominator satisfies
		$$ \sum_{j=1}^d |DAy|_j^2 \geq \alpha \sum_{j=1}^d |Ay|_j^2.$$
		As a result,
		$$ \max_{y\in \mathbb{C}^{r}} \frac{|DAy|_i^2}{\sum_{j=1}^d |DAy|_j^2} = \max_{y \in \mathbb{C}^{r}} \frac{D_{ii}^2|Ay|_i^2}{\sum_{j=1}^d |DAy|_j^2} \leq \max_{y\in \mathbb{C}^{r}} \frac{\beta|Ay|_i^2}{\alpha\sum_{j=1}^d |Ay|_j^2}.$$
		This proves that $ \tau_i(DA) \leq \frac{\beta}{\alpha}\tau_i(A)$ as desired.
	\end{proof}
	
	\begin{lemma}\label{levscore_submatrix}
		Let $A \in \mathbb{C}^{d \times r}$ be a matrix, and let $B \in \mathbb{C}^{d' \times r}$ be a matrix formed by taking any subset $R \subseteq [1,\dots,d]$ of the rows of $A$. Then the leverage score of any row $b_i$ in $B$ is at least the leverage score of row $b_i$ in $A$.
	\end{lemma}
	
	\begin{proof}
		Let $b_i$ correspond to the $i^{th}$ row of $A$ and the $i'^{th}$ row of $B$. Then by the definition of leverage scores we have the following.
		\begin{align*}
			\tau_i(A) &= \max_{y} \frac{|Ay|_i^2}{\sum_{j=1}^d |Ay|_j^2} \\
			&\leq \max_{y} \frac{|Ay|_i^2}{\sum_{j \in R} |Ay|_j^2} \\
			&= \max_{y} \frac{|By|_{i'}^2}{\sum_{j=1}^{d'} |By|_j^2} \\
			&= \tau_{i'}(B).
		\end{align*}
		This completes the proof of the lemma.
	\end{proof}
	\begin{lemma}\label{lemma:lev_score_col_span}
		Let $A \in \mathbb{C}^{d \times r}$ be a matrix, and let $B \in \mathbb{C}^{d \times r'}$ be a matrix formed by taking $r'$ linear combinations of the columns of $A$ (i.e. there is some matrix $M \in \mathbb{C}^{r \times r'}$ such that $B = AM$). Then the $i^{th}$ leverage score of $B$, $\tau_i(B)$, is at most the $i^{th}$ leverage score of $A$, $\tau_i(A)$. 
	\end{lemma}
	\begin{proof}
		We again rely on the maximization characterization of leverage scores: 
		\begin{align*}
			\tau_i(B) &= \max_{y} \frac{|By|_i^2}{\sum_{j=1}^d |By|_j^2} \\
			&= \max_{y} \frac{|AMy|_i^2}{\sum_{j=1}^d |AMy|_j^2} \\
			&\leq \max_{z} \frac{|Az|_i^2}{\sum_{j=1}^d |Az|_j^2} \\
			&= \tau_i(A).
		\end{align*}
		where the second inequality follows from the fact that $\{ My | y \in \mathbb{C}^{r'} \} \subseteq \{ z | z \in \mathbb{C}^r\}$.
	\end{proof}
	Equipped with these lemmas, the high level strategy for bounding the leverage scores of $WF_SM_S$ is as follows. By \ref{lemma:lev_score_col_span}, it suffices to bound the leverage scores of $WF_S$. To do so, first we will bucket the rows of $WF_S$ into submatrices, such that within each submatrix, the weights vary by at most a constant factor. We will then be able to bound the leverage scores of each submatrix using Lemma \ref{levscore_reweight}, and apply Lemma \ref{levscore_submatrix} to ensure that these upper bounds remain valid upper bounds for the entire matrix $WF_S$.
	
	Formally, begin by dividing $WF_S$ into submatrices as follows.
	\begin{definition}
		For all $i\in [\log d]$, let $W_i$ 
		be the submatrix of $W$ consisting of rows with indices in $[d(1-1/2^{i-1})+1, d(1-1/2^{i})]$. For the edge case $i=(\log d) + 1$, let $W_i$ consist of only the last row of $WF_S$. For convenience, let $\mathcal{R}_i$ denote the index set of rows corresponding to $W_i$, and let $r_i = |\mathcal{R}_i|$ denote the number of rows in $W_i$. Note also that $r_i \leq \frac{d}{2^i}$ for $i \in [\log d]$, and $r_{\log d + 1} = 1$. 
	\end{definition}
	Then we can write $WF_S$ as $WF_S = [(W_1F_S)^{T};\ldots; (W_{\log d +1}F_S)^{T}]^{T}$. Now consider any $i\in [\log d +1]$ and $W_i F_S$. Since $W_i$ is diagonal and $d/2^{i-1}\leq (W_{i})_{j,j}^2 \leq 2d/2^{i-1}$ for all $j\in [r_i]$, applying Lemma \ref{levscore_reweight} with $\beta = 2\alpha = 2d/2^{i}$ we get the following for all $i\in [\log d + 1]$:
	\begin{equation*}
		\tau_{j}(W_i F_S) \leq 2 \tau_{j}(F_{S,i}) \quad \forall j\in [r_i].
	\end{equation*}
	Here, $F_{S,i}\in \mathbb{C}^{r_i\times s}$ is the matrix consisting of all rows of $F_{S}$ in the index set $\mathcal{R}_i$. 
	Note that the column span of $F_{S,i}$ and the first $r_i$ rows of $F_S$ is identical. This is because the $j^{th}$ column of $F_{S,i}$ is 
	a constant times the $j^{th}$ column of the matrix formed by considering the first $r_i$ rows of $F_S$. Since the first $r_i$ rows of $F_S$ forms a $r_i \times r$ Fourier matrix, their leverage scores are then identical as well. Finally, we appeal to Corollary C.2 of \cite{eldar2020toeplitz}, restated below:
	\begin{corollary}[Corollary C.2 of \cite{eldar2020toeplitz}]\label{corr:C2_eldar}
		For any positive integers $d$ and $s \leq d$, there is an explicit set of values $\wt{\tau}_1^{(s)},\dots,\wt{\tau}_d^{(s)} \in (0,1])$ such that, for any Fourier matrix $F_S \in \mathbb{C}^{d \times s}$ with leverage scores $\tau_1,\dots,\tau_d$,
		\begin{align*}
			&\forall j, \wt{\tau}_j^{(s)} \geq \tau_j .\\
			& \sum_{j=1}^d \wt{\tau}_j^{(s)} = O(s \log s).
		\end{align*}
	\end{corollary}
	
	Thus by Corollary \ref{corr:C2_eldar} (C.2 of \cite{eldar2020toeplitz}), we easily obtain the following claim.
	\begin{claim}\label{levscore_upperbound}
		For any $i\in [\log (d/r)]$ and any $j\in [r_i]$ define $\wt{\tau}_{j,i}$ as follows.
		\begin{equation*}
			\wt{\tau}_{j,i} = \min\left(1,\frac{r}{\min(j,d/2^i+1-j)},\frac{O(r^6\log^3 (r+1))}{(d/2^i)}\right).
		\end{equation*}
		For any $\log(d/r)<i\leq \log(d)$ and $j\in [r_i]$ let $\wt{\tau}_{j,i}=1$. Then we have the following.
		\begin{enumerate}
			\item $\tau_j(W_iF_S) \leq \wt{\tau}_{j,i}$ for all $i\in [\log d]$ and $j\in [r_i]$.
			\item $\sum_{j=1}^{d/2^i} \wt{\tau}_{j} \leq O(r\log r)$ for all $i\in [\log d]$.
		\end{enumerate}
	\end{claim}
	Equipped with these tools, we can now easily finish the proof of Lemma \ref{leveragescorebound}.
	\begin{proof}
		Define $\wt{\tau}_j$ values for $j\in [d]$ as follows.
		
		\begin{equation}
			\wt{\tau}_{j} = \wt{\tau}_{i,j'}.
		\end{equation}
		where $i,j'$ are chosen such that $d(1-1/2^{i-1})+1 \leq j\leq d(1-1/2^{i})$ and $j = d(1-1/2^{i-1})+1 + j'$, and $\wt{\tau}_{i,j}$ are obtained from Claim \ref{levscore_upperbound}. Finally, using Lemma \ref{levscore_submatrix} we can easily conclude that these upper bounds also serve as upper bounds on the leverage scores of $WF_S$, and by \ref{lemma:lev_score_col_span}, upper bounds on the leverage scores of $WF_SM_S$ as well. This completes the proof of Lemma \ref{leveragescorebound}.
	\end{proof}


\section{Conclusion.}
In this paper, we study the design of sublinear algorithms for obtaining low-rank approximations of positive semidefinite Toeplitz matrices. Given query access to any such matrix $T\in \mathbb{R}^{d\times d}$, one can trivially reconstruct it exactly by reading its first column, i.e. by reading $d$ entries. Our main result is that for any $k,\epsilon,\delta$, there exists a symmetric Toeplitz $\wt{T}$ of rank $\wt{O}((k/\epsilon)\log(1/\delta))$ satisfying
\begin{equation*}
	\|T-\wt{T}\|_F \leq (1+\epsilon)\|T-T_k\|_F + \delta \|T\|_F,
\end{equation*}
where $T_k = \underset{B:rank(B)\leq k}{\argmin}\|T-B\|_F$ is the best rank-$k$ approximation $T$ in the Frobenius norm. Surprisingly, such an existence result -- that there exists a near optimal low-rank approximation to $T$ which is itself Toeplitz -- was not known before. We obtain this result by proving new results about the low rank structure of off-grid Fourier matrices, which we believe to be of independent interest. We also present an algorithm that reconstructs such a $\wt{T}$ by reading only $\wt{O}(k^2 \log(1/\delta)\poly(1/\epsilon))$ entries of $T$, beating the trivial bound of $d$ queries and thus achieving sublinear query complexity. We now present some of the main open problems raised by this work. 
\begin{enumerate}
	\item Is the additive error term in Theorem \ref{frobeniusexistence} necessary? Also, what is the minimum rank of required to achieve the guarantee of Theorem \ref{frobeniusexistence}?
	\item Is it possible to design a sublinear time low-rank approximation algorithm that recovers a $\wt{T}$ achieving the guarantee of Theorem \ref{sublinearqueryalgo}?
	\item Is it possible to design an algorithm with sublinear query complexity, or even sublinear run-time, that can recover a $\wt T$ satisfying the spectral norm low-rank approximation guarantee of Theorem \ref{spectralexistence}?
	\item Can the existence of a structure preserving low-rank approximation be proven for non-PSD Toeplitz matrices, similar to Theorem \ref{frobeniusexistence}? Can a sublinear query and sublinear time algorithm be designed to recover near optimal low-rank approximations to non-PSD Toeplitz matrices?
\end{enumerate}
\section{Acknowledgements.}
Michael Kapralov and Mikhail Makarov's work is supported by the European Research Council (ERC) under the European Union’s Horizon 2020 research and innovation programme (grant agreement No 759471). Hannah Lawrence is supported by the Fannie and John Hertz
Foundation and the National Science Foundation Graduate Research Fellowship under Grant No.
1745302. Cameron Musco's work is supported by an Adobe Research grant, a Google Research Scholar Award, and NSF Grants No. 2046235 and
No. 1763618.

\bibliographystyle{plain} 
\bibliography{refs}

\end{document}